\documentclass{article}
\usepackage{color, xcolor, colortbl}
\usepackage{graphicx}
\usepackage{geometry}
\usepackage{amsthm,amsmath,amssymb,amsfonts}
\usepackage{algorithm}
\usepackage{algorithmic}
\usepackage{dcolumn}
\usepackage{epstopdf}
\usepackage{bm}
\usepackage[caption=false]{subfig}
\usepackage{appendix}
\usepackage{multirow}
\usepackage{braket}
\usepackage[english]{babel}
\usepackage{hyperref}
\usepackage[capitalize]{cleveref}
\usepackage{xpatch}
\usepackage{tikz}
\usepackage{adjustbox}
\usepackage{xspace}

\renewcommand{\Re}{\operatorname{Re}}
\renewcommand{\Im}{\operatorname{Im}}

\newcommand{\poly}{\operatorname{poly}}

\let\originalleft\left
\let\originalright\right
\renewcommand{\left}{\mathopen{}\mathclose\bgroup\originalleft}
\renewcommand{\right}{\aftergroup\egroup\originalright}

\newcommand{\mc}[1]{\mathcal{#1}}

\newcommand{\wt}[1]{\widetilde{#1}}

\newcommand{\abs}[1]{\left\lvert#1\right\rvert}
\newcommand{\norm}[1]{\left\lVert#1\right\rVert}

\newcommand{\ud}{\,\mathrm{d}}
\newcommand{\Or}{\mathcal{O}}

\newcommand{\RR}{\mathbb{R}}
\newcommand{\CC}{\mathbb{C}}

\newtheorem{thm}{\protect\theoremname}
\newtheorem{lem}[thm]{\protect\lemmaname}

\newtheorem{cor}[thm]{\protect\corollaryname}

\newtheorem{res}[thm]{Result}

\providecommand{\definitionname}{Definition}
\providecommand{\assumptionname}{Assumption}
\providecommand{\corollaryname}{Corollary}
\providecommand{\lemmaname}{Lemma}
\providecommand{\propositionname}{Proposition}
\providecommand{\remarkname}{Remark}
\providecommand{\theoremname}{Theorem}

\usetikzlibrary{fit}
\tikzset{%
  highlight/.style={rectangle,rounded corners,fill=blue!15,draw,fill opacity=0.3,thick,inner sep=0pt}
}

\usepackage{authblk}

\usepackage[normalem]{ulem}

\begin{document}

\title{Laplace transform based quantum eigenvalue transformation via linear combination of Hamiltonian simulation}
\author[1,2]{Dong An}
\author[2,3]{Andrew M.\ Childs}
\author[4,5,6]{Lin Lin}
\author[7,8]{Lexing Ying}

\affil[1]{Beijing International Center for Mathematical Research, Peking University}
\affil[2]{Joint Center for Quantum Information and Computer Science \authorcr University of Maryland, College Park}
\affil[3]{Department of Computer Science and
Institute for Advanced Computer Studies \authorcr University of Maryland, College Park}
\affil[4]{Department of Mathematics, University of California, Berkeley}
\affil[5]{Applied Mathematics and Computational Research Division \authorcr Lawrence Berkeley National Laboratory, Berkeley}
\affil[6]{Challenge Institute for Quantum Computation, University of California, Berkeley}
\affil[7]{Department of Mathematics, Stanford University}
\affil[8]{Institute for Computational and Mathematical Engineering, Stanford University}

\date{\today}
\maketitle

\begin{abstract}
Eigenvalue transformations, which include solving time-dependent differential equations as a special case, have a wide range of applications in scientific and engineering computation.
While quantum algorithms for singular value transformations are well studied, eigenvalue transformations are distinct,  especially for non-normal matrices.
We propose an efficient quantum algorithm for performing a class of eigenvalue transformations that can be expressed as a certain type of matrix Laplace transformation. This allows us to significantly extend the recently developed linear combination of Hamiltonian simulation (LCHS) method [An, Liu, Lin, Phys. Rev. Lett. 131, 150603, 2023; An, Childs, Lin, arXiv:2312.03916] to represent a wider class of eigenvalue transformations, such as powers of the matrix inverse, $A^{-k}$, and the exponential of the matrix inverse, $e^{-A^{-1}}$. The latter can be interpreted as the solution of a mass-matrix differential equation of the form $A u'(t)=-u(t)$. We demonstrate that our eigenvalue transformation approach can solve this problem without explicitly inverting $A$, reducing the computational complexity.

\end{abstract}

\tableofcontents

\section{Introduction}

Quantum computers are expected to solve certain computational problems much more efficiently than classical computers, such as factoring large integers~\cite{Shor1997} and simulating the dynamics of quantum systems~\cite{BerryChildsKothari2015,LowChuang2017,GilyenSuLowEtAl2019,ChildsSuTranEtAl2021}. 
Large scale and high dimensionality appear ubiquitously in scientific and engineering computation and have posed significant challenges for classical computers, motivating the development of efficient quantum algorithms for such problems.

Many scientific computing tasks can be expressed as eigenvalue transformations. 
Suppose $A\in\CC^{N\times N}$ is diagonalizable as $A=\mc{V} D \mc{V}^{-1}$, where $\mc{V}$ is an invertible matrix and $D$ is a diagonal matrix. Let $h:\CC\to \CC$ be a function defined on the spectrum of $A$. The task of quantum eigenvalue transformation is to encode the matrix 
\begin{equation}
h(A):=\mc{V} h(D) \mc{V}^{-1}
\label{eqn:eigenvalue_transformation}
\end{equation}
using a unitary matrix that can be efficiently implemented on a quantum computer, or to prepare a quantum state proportional to $h(A) \ket{\psi}$ for an input state $\ket{\psi}$. 
The most well-known instance of this problem involves the linear time-independent differential equation
\begin{equation}\label{eqn:firstorder_hom}
\frac{\ud u(t)}{\ud t} = -A u(t), \quad 0\le t\le T.
\end{equation}
The solution can be expressed as a matrix function $u(T)=e^{-AT} u(0)$, so~\cref{eqn:firstorder_hom} can be solved by performing the matrix exponential $e^{-AT}=\mc{V} h(D) \mc{V}^{-1}$, which is an eigenvalue transformation with $h(z)=e^{-zT}$.

Certain eigenvalue transformation problems, such as the matrix inverse $A\mapsto A^{-1}$,  can be formulated as singular value transformation (SVT) problems.\footnote{Given a matrix $A$ with singular value decomposition $A=V\Sigma W^{\dag}$, the singular value transformation with function $h$ produces the matrix $V h(\Sigma) W^\dag$. The matrix inverse can be implemented by applying a singular value transformation on the hermitian conjugate of $A$ with function $h(x)=1/x$, since $(A^{\dag})^{-1}=V \Sigma^{-1} W^{\dag}$. 
} 
This enables the usage of efficient quantum algorithms such as the quantum singular value transformation (QSVT)~\cite{GilyenSuLowEtAl2019}. However, most eigenvalue transformations, including matrix exponentials for general matrices, cannot be reformulated in such a fashion. Other examples include powers of the matrix inverse $A\mapsto A^{-k}$ ($k\in\RR_{+}$), the exponential of the matrix inverse $A\mapsto e^{-A^{-1}}$, and the matrix sign function $A\mapsto {\rm sgn}(A-\mu I)$ for a given $\mu \in \RR$, to name a few.
Two primary resources should be accounted for in the cost of preparing $h(A)\ket{\psi}$: the number of queries to the matrix $A$ and the number of queries to the input state $\ket{\psi}$.

\subsection{Related works}

We are aware of two strategies for implementing general quantum eigenvalue transformations. The first strategy is the contour integral-based formulation \cite{TakahiraOhashiSogabeEtAl2021,TongAnWiebeEtAl2021}. When the function $h$ is analytic within a region enclosing all eigenvalues of $A$ in the complex plane, the contour integral technique can be described as an integral over matrix inverses. After discretization of the contour, the eigenvalue transformation $h(A)$ can be implemented using a linear combination of matrix inverses, where each matrix inverse can be implemented by solving a quantum linear system problem (QLSP), for which many algorithms including QSVT are available~\cite{HarrowHassidimLloyd2009,Ambainis2012,ChildsKothariSomma2017,SubasiSommaOrsucci2019,LinTong2020,AnLin2022,CostaAnSandersEtAl2022,Dalzell2024,LowSu2024quantumlinearalgorithmoptimal}.

The second strategy for a general eigenvalue transformation uses Taylor expansion of $h(z)$ into a linear combination of monomials, and implements each matrix monomial using products of block encodings~\cite{GilyenSuLowEtAl2019}. 
However, the monomial coefficients can be arbitrarily large even for well-behaved $h(z)$ such as Chebyshev polynomials, making the Taylor expansion approach highly inefficient. 
The recent development of quantum eigenvalue transformation (QEVT) by Low and Su \cite{LowSu2024} overcomes this problem via a history state polynomial (HSP) using \textit{stable} polynomial expansions. For instance, for a non-normal matrix $A$ with real eigenvalues within $[-1,1]$, this approach introduces a Chebyshev history state to encode all Chebyshev polynomials up to a specified order simultaneously. When the eigenvalues are not real but lie within the unit circle, the method constructs a Faber polynomial approximation via a Faber polynomial history state. Using these history state polynomials, the eigenvalue transformation can be implemented again by solving a QLSP.

Both the contour integral method and the history state polynomial method described above can be used to implement the matrix exponential $e^{-AT}$. However, the connection between the matrix exponential and the differential equation \cref{eqn:firstorder_hom} allows us to employ special techniques for this particular matrix eigenvalue problem.
One type of algorithm, called the linear system approach \cite{Berry2014,BerryChildsOstranderEtAl2017,ChildsLiu2020,Krovi2022,BerryCosta2022}, involves discretizing the time interval $[0, T]$ into $L$ shorter intervals of length $\Delta t = T/L$. This method uses a short-time propagator to propagate from one time interval to another and formulates a dilated linear system to encode the history state across the interval $[0, T]$. While the best existing algorithm of this type~\cite{BerryCosta2022} may offer a near-optimal number of queries to $A$, it does not optimize the initial state preparation cost since the construction of such algorithm uses the quantum linear system algorithm in~\cite{CostaAnSandersEtAl2022} as a subroutine, which requires multiple copies of the initial state. 

The time-marching method \cite{FangLinTong2023} also employs a short-time propagator as in the linear system approach and directly encodes the matrix exponential $e^{-AT}$ by performing repeated postselections of the short-time propagators. 
It can achieve optimal state preparation cost, but the complexity with respect to the number of queries to $A$ is sub-optimal. 
Furthermore, to our knowledge, neither the linear system approach nor the time-marching approach is well-suited for general eigenvalue transformations since they heavily rely on discretizing differential equations in the form of~\cref{eqn:firstorder_hom}. 

The linear combination of Hamiltonian simulation (LCHS) method \cite{AnLiuLin2023} significantly simplifies the process of constructing a block encoding $e^{-AT}$ when the matrix $-A$ is dissipative, in the sense that the Hermitian part $L=(A+A^{\dag})/2$ of the matrix $A$ is a positive semidefinite matrix. The core of LCHS is an identity expressing $e^{-AT}$ as a linear combination of a continuously parameterized family of Hamiltonian simulation problems, which can then be implemented using any quantum Hamiltonian simulation algorithm without resorting to specific short-time integrators. 
It also achieves optimal state preparation cost for matrix exponentiation. A closely related approach for implementing $e^{-AT}$ is the Schr\"odingerisation method~\cite{JinLiuYu2022}, which converts a non-unitary differential equation into a dilated Schr\"odinger equation with an additional momentum dimension, subject to specific initial conditions.
The approaches in \cite{AnLiuLin2023,JinLiuYu2022} are first-order schemes, and the matrix query complexity scales linearly with $1/\epsilon$, where $\epsilon$ is the target precision.
Recently, LCHS has been generalized into a family of identities for expressing $e^{-AT}$ with exponentially improved accuracy~\cite{AnChildsLin2023}. This led to the first quantum algorithm to block-encode matrix exponentials with both optimal state preparation cost and near-optimal scaling in matrix queries across all parameters, including the target precision $\epsilon$ and simulation time $T$. 

The central question of this work is as follows:

\begin{quote} \textit{Can the linear combination of Hamiltonian simulation method be generalized to efficiently represent matrix eigenvalue transformations beyond matrix exponentiation?} \end{quote}

\subsection{Contribution and main idea}

In this work, we propose the following type of eigenvalue transformation.
Let $h(z)$ be a function represented as the Laplace transform of a function ${g}(t) \in L^1(\RR_+)$: 
\begin{equation}\label{eqn:laplace_forward}
h(z) = \int_{0}^{\infty} {g}(t) e^{-zt} \ud t, \quad \Re z\ge 0.
\end{equation}
If $-A$ is dissipative (i.e., all the eigenvalues of $A$ are in the right half-plane), then the eigenvalue transformation $h(A)$ takes the form 
\begin{equation}\label{eqn:matrix_laplace}
h(A)=\int_{0}^{\infty} {g}(t) e^{-At} \ud t.
\end{equation}
This is a linear combination of matrix exponentials, and each matrix exponential $e^{-At}$ can be further written as a linear combination of Hamiltonian simulation problems using the technique in~\cite{AnChildsLin2023}.
As a result, $h(A)$ is formulated as a double integral of continuously parameterized Hamiltonian simulation problems. 
We call this the \textit{Laplace transform based linear combination of Hamiltonian simulations} (Lap-LCHS). 
\begin{res}[Informal version of~\cref{thm:LCHS_hardy}]\label{res:Lap_LCHS}
    Let $A=L+iH$, where $L=\frac{A+A^{\dag}}{2} \succeq 0$ and $H=\frac{A-A^{\dag}}{2i}$ are both Hermitian matrices. 
    For a function $h(z) = \int_{0}^{\infty} {g}(t) e^{-zt} \ud t$ that can be expressed by a Laplace transformation with $g(t)$ being its inverse Laplace transform, we have 
    \begin{equation}\label{eqn:Lap_LCHS_intro}
    h(A) = \int_0^{\infty} \int_{\mathbb{R}} \frac{f(k){g}(t)}{ 1-ik} e^{-it (kL+H)} \ud k \ud t. 
    \end{equation}
    Here\footnote{In fact there exists a large family of kernel functions $f(k)$ that can make~\cref{res:Lap_LCHS} hold (see~\cref{thm:LCHS_hardy}). The one specified here is nearly optimal in the sense that it achieves almost the fastest poossible asymptotic decay along the real axis~\cite[Proposition 7]{AnChildsLin2023}.} $f(k) = \frac{1}{2\pi e^{-2^\beta} e^{(1+iz)^{\beta}} }$, $\beta \in (0,1)$.
\end{res}

Using~\cref{res:Lap_LCHS}, we construct an efficient quantum algorithm to implement the matrix function $h(A)$. 
We truncate and discretize the double integral in~\cref{eqn:Lap_LCHS_intro}, approximating $h(A)$ by a weighted sum of Hamiltonian simulation operators $e^{-it(kL+H)}$ for some specific choices of $t$ and $k$. 
Then our quantum Lap-LCHS algorithm implements each Hamiltonian simulation operator by the optimal time-independent Hamiltonian simulation algorithm based on QSVT~\cite{GilyenSuLowEtAl2019}, and performs the weighted sum using the linear combination of unitaries (LCU) technique~\cite{ChildsWiebe2012,ChildsKothariSomma2017}. 

We establish a detailed complexity estimate for the Lap-LCHS algorithm. 
Here we present the cost for preparing the normalized state $\frac{h(A)\ket{\psi}}{\norm{h(A)\ket{\psi}}}$. We also analyze the cost of block encoding $h(A)$ in \cref{thm:complexity_block_encoding}. 

\begin{res}[Informal version of~\cref{cor:complexity_state}]\label{res:Lap_LCHS_cost}
    Let $A=L+iH$, where $L=\frac{A+A^{\dag}}{2}\succeq 0$ and $H=\frac{A-A^{\dag}}{2i}$ are both Hermitian matrices. 
    For a given function $h(z)$, we can prepare the state $h(A)\ket{\psi}/\norm{h(A)\ket{\psi}}$ with error at most $\epsilon$ using $\widetilde{\Or}\left( \frac{\norm{{g}}_{L^1(\RR_+)}}{\norm{h(A) \ket{\psi}}} \left(\alpha_A T \left(\log\left(\frac{1}{\epsilon}\right)\right)^{1+o(1)} \right) \right)$ queries to the matrix $A$ and $\Or\left( \frac{\norm{{g}}_{L^1(\RR_+)}}{\norm{h(A) \ket{\psi}}} \right)$ queries to the input state $\ket{\psi}$. 
    Here $\alpha_A \geq \norm{A}$, $g(t)$ is the inverse Laplace transform of $h(z)$, and $T$ is a truncation parameter such that $\norm{{g}}_{L^1((T,\infty))} = \Or(\epsilon \norm{h(A)\ket{\psi}})$. 
\end{res}

In~\cref{res:Lap_LCHS_cost}, the query complexity to the matrix $A$ depends linearly on the truncation parameter $T$, which depends implicitly on the error $\epsilon$. Specifically, $T$ is chosen such that the integral $\int_T^{\infty} |g(t)| \ud t $ is bounded by $\Or(\epsilon \norm{h(A)\ket{\psi}})$. This may introduce an additional polynomial computational overhead with respect to $\epsilon$ if $g(t)$ decays only polynomially as $t \rightarrow \infty$. Consequently, the overall matrix query complexity with respect to the error should be carefully analyzed on a case-by-case basis.

The contour integral method, QEVT with history state polynomial, and Lap-LCHS each have different ranges of applicability where they are most efficient, so it is difficult to provide a succinct general comparison of these methods. As a rule of thumb, the contour integral method may require a part of the contour to be in the left half of the complex plane, so the cost \textit{can} grow exponentially in $T$ for $h(z)$ of the form in \cref{eqn:laplace_forward}.
According to \cite[Proposition 34]{LowSu2024}, QEVT also requires the Hermitian part $L$ of the matrix $A$ to be positive semi-definite, just as Lap-LCHS does. However, it may be possible to implement a broader range of transformations using QEVT,  as it only requires the complex analyticity of the function 
$h(z)$ over a bounded domain and does not rely on any assumptions regarding the existence of its inverse Laplace transform.
Both Lap-LCHS and QEVT can achieve near-optimal complexity with respect to the number of queries to $A$, but Lap-LCHS has lower state preparation cost. 
Specifically, Lap-LCHS achieves the optimal state preparation cost, while the cost of QEVT involves additional logarithmic overhead in terms of the polynomial degree, even when combined with the recently developed tunable variable time amplitude amplification (VTAA) and block preconditioning techniques \cite{LowSu2024quantumlinearalgorithmoptimal}. 
Lap-LCHS is also simpler and might be more feasible in the early fault-tolerant regime by implementing the linear combination stochastically~\cite{WangMcArdleBerta2024,Chakraborty2024}, while it remains unclear how to implement QEVT with near-optimal scaling on intermediate-term quantum computers due to the complexity of VTAA.

To demonstrate the computational advantage of Lap-LCHS over other methods for performing eigenvalue transformations, we consider several specific instances. 
For all problems considered in this work, we find that Lap-LCHS achieves the best known complexity both in terms of the number of queries to $A$ and the initial state preparation oracle. 

The first example is the matrix function 
\begin{equation}
h(z)=\frac{1}{ z}(1-e^{-z T })=\int_0^T e^{-zs} \ud s,
\end{equation} 
so that $g(s)$ is an indicator function on the interval $[0,T]$. This corresponds to the simulation of a special inhomogeneous differential equation 
\begin{equation}\label{eqn:firstorder_inhom}
\frac{\ud u(t)}{\ud t} = -A u(t) + \ket{\psi}, \quad u(0)=0, \quad 0\le t\le T. 
\end{equation}
Generalizing the homogeneous equation in~\cref{eqn:firstorder_hom}, the inhomogeneous equation can model an external driving force. 
We apply the Lap-LCHS algorithm and find that in this case it is equivalent to combining the standard LCHS and variation of constants.
This example is presented in \cref{sec:linear_inhom}.

The second example is the power of matrix inverses  
\begin{equation}
   h(z) = (\eta + z)^{-p}
\end{equation}
for some $\eta>0,p>0$.
This transformation is a generalization of the standard matrix inverse, and the power $p$ does not need to be an integer. Fractional powers of matrices arise in certain Markov chain models and fractional differential equations \cite{HighamLin2011}. 
Here the matrix function $h(A) = (\eta I + A)^{-p}$ involves a shift by $\eta$. Notice that we can implement $A^{-p}$ if all the eigenvalues of $A+A^{\dagger}$ are positive. Then we can choose $\eta$ to be the smallest eigenvalue of $(A+A^{\dagger})/2$ and consider the matrix function $h(A-\eta I)$. 
This example is presented in \cref{sec:power_inverse}. 

The third example is the simulation of a linear differential equation with a non-diagonal mass matrix:
\begin{equation}\label{eqn:firstorder_mass}
A\frac{\ud u(t)}{\ud t} = - u(t), \quad 0\le t\le T.
\end{equation}
We consider two types of initial conditions, namely 
\begin{equation}
    u(0) = A^{-1} u_0, \quad \mbox{and}\quad u(0) = u_0.
\end{equation}
In both cases, we are given an oracle for preparing the state $\ket{u_0} = u_0/\norm{u_0}$. 
\Cref{eqn:firstorder_mass} has applications to evolutionary partial differential equations with time-space mixed derivatives in the form of 
\begin{equation}
    \partial_t \mathcal{L}_x u(t,x) = - u(t,x), 
\end{equation}
where $\mathcal{L}_x = \sum_{k=0}^{m} a_k \partial_x^k$ is a differential operator with respect to the spatial variable $x$. 
We can discretize the spatial variable to obtain a semi-discretized partial differential equation in the form of~\cref{eqn:firstorder_mass}, where $A$ is the discrete version of $\mathcal{L}_x$. 

We apply the Lap-LCHS algorithm to~\cref{eqn:firstorder_mass} and analyze its complexity in~\cref{sec:firstorder_mass}. 
Remarkably, Lap-LCHS can solve~\cref{eqn:firstorder_mass} (i.e., implement the transformation $e^{-T A^{-1}}$) without explicitly inverting the matrix $A$, and achieves the best known complexity both in terms of the number of queries to $A$ and the initial state preparation oracle. 

The last example is the second-order differential equation
\begin{equation}
\frac{\ud^2 u}{\ud t^2}=Au, \quad u(0) = \ket{u_0}, \quad u'(0)=v_0,
\end{equation}
which has many applications such as wave equations. 
These second-order differential equations have a branch of the solution that grows in time. 
We apply Lap-LCHS to simulate the \textit{decaying} branch of the solution. Note that if we first convert the second-order differential equation into a set of first-order differential equations and apply the matrix exponential solver, the cost will be dominated by the growing solution. 
Similarly to the first-order case, we may also consider the second-order differential equation with a non-diagonal mass matrix
\begin{equation}
A\frac{\ud^2 u}{\ud t^2}=u, \quad u(0) = \ket{u_0}, \quad u'(0)=v_0.
\end{equation}
The results are presented in \cref{sec:secondorder}.

\subsection{Discussion}
 
Despite significant recent advances in quantum algorithms for linear algebraic transformations, research on eigenvalue transformations for non-normal matrices remains in a relatively early stage. Consider for example Hamiltonian simulation, which is obviously an eigenvalue transformation problem. However, an efficient implementation of $e^{-iHt}$ using the contour integral method or the history state polynomial method is far from obvious, while Lap-LCHS simply reduces to the Hamiltonian simulation problem itself by construction. To our knowledge, when $-A$ satisfies the dissipative condition, Lap-LCHS obtains the lowest  complexity among existing approaches for eigenvalue transformations. However, a central unifying structure akin to the qubitization framework that underpins QSVT and prior developments in singular value transformations has yet to be clearly identified for the eigenvalue transformation of non-normal matrices. Establishing such a unifying perspective could greatly enhance our understanding of quantum algorithms for matrix computation.

\subsection{Organization}

The rest of this paper is organized as follows. 
First, in~\cref{sec:prelim}, we review two existing approaches for special eigenvalue transformations: QSVT for Hermitian matrix functions and LCHS for matrix exponentials. 
In~\cref{sec:Lap-LCHS}, we establish the framework of Lap-LCHS, discuss the construction of our quantum algorithm, and present its complexity analysis. 
Then we discuss the applications of the Lap-LCHS algorithm in~\cref{sec:applications}. 

\subsection*{Acknowledgements}

DA acknowledges support from the Department of Defense through the Hartree Postdoctoral Fellowship at QuICS. 
AMC acknowledges support from the Department of Energy, Office of Science, Office of Advanced Scientific Computing Research, Accelerated Research in Quantum Computing program, and from the National Science Foundation (NSF) through QLCI grant OMA-2120757.
LL is a Simons investigator in Mathematics, and acknowledges support from the Challenge Institute for Quantum Computation (CIQC) funded by NSF through grant number OMA-2016245.
LY acknowledges support from NSF grant DMS-2208163. LL and LY also acknowledge support from the U.S. Department of Energy, Office of Science, Accelerated Research in Quantum Computing Centers, Quantum Utility through Advanced Computational Quantum Algorithms, grant no. DE-SC0025572.

\section{Preliminaries}\label{sec:prelim}

\subsection{Singular value transformation and eigenvalue transformation of Hermitian matrices}

The quantum singular value transformation (QSVT) is a powerful tool for designing quantum algorithms.
Given a square matrix $A\in \CC^{N\times N}$ with operator norm $\norm{A}\le 1$, its singular value decomposition can be expressed as $A=W\Sigma V^{\dag}$. Here $W,V$ are $N\times N$ unitary matrices, and $\Sigma$ is a diagonal matrix with diagonal entries in $[0,1]$. QSVT enables an efficient implementation of the singular value transformation for a broad class of polynomial $h:[-1,1]\to \CC$ of definite parity (i.e., $h$ is even or odd), expressed as
\begin{equation}
A\mapsto h^{\mathrm{SV}}(A):=\begin{cases}
W h(\Sigma)V^{\dag}, & \text{if } h \text{ is odd},\\
V h(\Sigma)V^{\dag}, & \text{if } h \text{ is even}.
\end{cases}
\end{equation}
This technique has unified a diverse range of tasks.

When $A$ is a Hermitian matrix and $A\succeq 0$, its eigenvalue decomposition and singular value decomposition coincide, and so do eigenvalue and singular value transformations of $A$. 

When $A$ is an indefinite Hermitian matrix, its eigenvalue decomposition is $A=\mc{V}D \mc{V}^{\dag}$ for some unitary matrix $\mc{V}$.  Its singular value decomposition can be written as $A=W\Sigma V^{\dag}$ with $W=\mc{V}\operatorname{sign}(D), V=\mc{V}$, and $\Sigma=\abs{D}$ is obtained by taking the absolute values of the diagonal elements in $D$. In this case, if $f$ is an odd function, then
\begin{equation}
f^{\mathrm{SV}}(A)=Wf(\Sigma)V^{\dag}=V f(\operatorname{sign}(D)\Sigma)V^{\dag}=\mc{V}f(D)\mc{V}^{\dag}=f(A).
\end{equation}
If $f$ is an even function, then
\begin{equation}
f^{\mathrm{SV}}(A)=Vf(\Sigma)V^{\dag}=\mc{V}f(D)\mc{V}^{\dag}=f(A).
\end{equation}
Therefore, as long as $f$ has definite parity, eigenvalue and singular value transformations of a Hermitian matrix $A$ are the same.

For a more general $A$, there is no longer a clear connection between eigenvalue transformations and singular value transformations.

\subsection{Laplace transform and Hardy functions}\label{sec:laplace_hardy}
 
The Laplace transform is closely connected to a class of functions called the Hardy space. 
 In the upper half-plane $\CC^+=\{z\in\CC \mid \Im z>0\}$, a function $f(z)$ belongs to the Hardy space $H^p(\CC^+)$ ($1\le p<\infty$) if $f(z)$ is holomorphic in $\CC^+$ and satisfies
 \begin{equation}
 \sup_{y > 0} \int_{-\infty}^{\infty} |f(x + iy)|^p \ud x < \infty.
 \end{equation}
 The Hardy space $H^\infty(\CC^+)$ is given by bounded holomorphic functions in $\CC^+$, i.e., 
 \begin{equation}
 \sup_{z \in \CC^+} |f(z)| < \infty.
 \end{equation}
 For any function $f\in H^p(\CC^+)$ with $1\le p<\infty$, the boundary value on the real axis can be obtained by taking the non-tangential limit as $z$ approaches the real axis~\cite[VI.C]{Koosis1998}. The limiting function, denoted by $\wt{f}(x)$, satisfies $\wt{f}\in L^p(\RR;\CC)$, and the Hardy function $f(z)$ can be reconstructed from the boundary value. 
 From the Hardy space $H^p(\CC^+)$ we may construct Hardy spaces restricted to other regions of the complex plane via conformal mappings. We introduce the notation of the upper and lower half-plane
 \begin{equation}
 \CC^+=\{z\in\CC \mid \Im z>0\}, \quad \CC^{-}=\{z\in\CC \mid \Im z<0\}
 \end{equation}
 as well as the left and right half-plane
 \begin{equation}
 \CC^{>}=\{z\in\CC \mid \Re z>0\}, \quad \CC^{<}=\{z\in\CC \mid \Re z<0\},
 \end{equation}
 respectively.
 Clearly $f(z)\in H^p(\CC^{>})$ if $f(-i z)\in  H^p(\CC^{+})$ and so on.

 Let $h(z)$ be the Laplace transform of ${g}(t)$ according to \cref{eqn:laplace_forward}.
 Since this integral will be used to construct an LCU representation, we assume
 ${g}\in L^1(\RR_+)$. Then 
 \begin{equation}
 \abs{h(z)}\le \int_{0}^{\infty} \abs{{g}(t)} \ud t =: \norm{{g}}_{L^1(\RR_+)}, \quad \Re z\ge 0. 
 \end{equation}
 Furthermore, $h(z)$ is a holomorphic function on the right half-plane $C_>$, and is continuous on $\overline{C_>}$. Therefore, the function $h(z)$ of interest in this work is in the Hardy space $H^\infty(C_>)$.

\subsection{Linear combination of Hamiltonian simulation for matrix exponentials}

The Cartesian decomposition of $A\in\CC^{N\times N}$ is~\cite[Chapter I]{Bhatia1997}
\begin{equation}\label{eqn:A_cartesian_1}
A = L + iH,
\end{equation}
where the Hermitian matrices
\begin{equation}\label{eqn:A_cartesian_2}
L = \frac{A+A^{\dagger}}{2}, \quad H = \frac{A-A^{\dagger}}{2i}
\end{equation}
are called the matrix real and imaginary part of $A$, respectively. If $L \succeq 0$, then $-A$ is called dissipative. In particular, if $A v=\lambda v$, then 
\begin{equation}
\frac12 (v^{\dag} Av+v^{\dag} A^{\dag} v)=\Re \lambda \norm{v}^2=v^{\dag} L v\ge 0.
\end{equation}
Therefore, $\Re \lambda \ge 0$ for any eigenvalue $\lambda$ of the matrix $A$.

The recent generalization of the LCHS formula expresses $e^{-A}$ as an integral of Hamiltonian simulation problems weighted by a kernel. 
The following result is a restatement of \cite[Theorem 5]{AnChildsLin2023} in terms of Hardy functions. 
\begin{thm}\label{thm:LCHS_improved}
    Let $f(z)$ be a function of $z \in \mathbb{C}$, such that 
    \begin{enumerate} 
        \item (Hardy) $f(z)\in H^1(\CC_{-})$ is continuous on $\overline{\CC_{-}}$, 
        \item (Normalization) $\int_{\mathbb{R}} \frac{f(k)}{ 1-ik} \ud k = 1$. 
    \end{enumerate}
    Consider the Cartesian decomposition of $A$ given in \cref{eqn:A_cartesian_1} and suppose $L\succeq 0$. Then for $t \geq 0$, 
    \begin{equation}
    e^{-At} = \int_{\mathbb{R}} \frac{f(k)}{ 1-ik} e^{-i t (kL+H) } \ud k. 
\end{equation}
\end{thm}

An asymptotically near-optimal choice of the kernel function is 
\begin{equation}\label{eqn:kernel_intro}
    f(z) = \frac{1}{2\pi e^{-2^\beta} e^{(1+iz)^{\beta}} }, \quad \beta \in (0,1).
\end{equation}
This function decays at a near-exponential rate of $e^{-c |k|^{\beta}}$ on the real axis. Furthermore, $\beta$ can be arbitrarily close to (but not equal to) $1$~\cite[Proposition 7]{AnChildsLin2023}. 

\section{Laplace transform based linear combination of Hamiltonian simulations}\label{sec:Lap-LCHS}

In this section, we present the Lap-LCHS method for quantum eigenvalue transformations. 
We first show the theoretical foundation of Lap-LCHS for representing quantum eigenvalue transformations in~\cref{sec:Lap-LCHS_formulation}. 
The corresponding quantum algorithm is constructed from the discretized version of the LCHS representation; the numerical discretization is discussed in~\cref{sec:Lap-LCHS_quadrature}. 
Then, we list the oracles we assume for the quantum algorithm in~\cref{sec:Lap-LCHS_oracles}, walk through the steps of the quantum algorithm in~\cref{sec:Lap-LCHS_algorthm}, and estiamte its complexity in~\cref{sec:Lap-LCHS_complexity}.

\subsection{Formulation}\label{sec:Lap-LCHS_formulation}

\cref{thm:LCHS_improved} shows that the matrix exponential function $h(z)=e^{-z}$ has an LCHS representation. 
For a more general $h(z)$, we can apply a matrix version of the Laplace transform given in~\cref{eqn:matrix_laplace} and represent this transforrm as a linear combination of matrix exponential functions, which in turn can be interpreted as a linear combination of Hamiltonian simulation problems. 
Specifically, we have the following theorem. 

\begin{thm}\label{thm:LCHS_hardy}
    Consider functions $f(z), {g}(t)$ satisfying 
    \begin{enumerate} 
        \item $f\in H^1(\CC_{-})$ is continuous on $\overline{\CC_{-}}$, 
        \item $\int_{\mathbb{R}} \frac{f(k)}{ 1-ik} \ud k = 1$, and
        \item ${g}\in L^1(\RR_+)$.
        \end{enumerate}
Define $h(z)$ according to the Laplace transform in \cref{eqn:laplace_forward}.
Consider the Cartesian decomposition of $A$ given in \cref{eqn:A_cartesian_1} and suppose $L\succeq 0$. Then 
\begin{equation}\label{eqn:lchs_matrixhardy}
    h(A) = \int_0^{\infty} \int_{\mathbb{R}} \frac{f(k){g}(t)}{ 1-ik} e^{-it (kL+H)} \ud k \ud t. 
\end{equation}
\end{thm}
\begin{proof}
By \cite[Theorem 6.2.27]{HornJohnson1991}, two continuous matrix functions are identical if and only if they are identical for all diagonalizable matrices in the domain.  Assuming $A=VDV^{-1}$ is diagonalizable, $h(A)=Vh(D)V^{-1}$ is well-defined. 
By \cref{thm:LCHS_improved}, the right-hand side is equal to
\begin{equation}
\int_0^{\infty} {g}(t) e^{-tA} \ud t=  V\left(\int_0^{\infty} {g}(t) e^{-tD} \ud t\right)V^{-1}=Vh(D)V^{-1}.
\end{equation}
This proves the theorem.
\end{proof}

In \cref{thm:LCHS_hardy}, we assume the inverse Laplace transform $g(t)$ to be in $L^1(\RR_+)$. Notice that validity of~\cref{eqn:lchs_matrixhardy} does not need $g \in L^1(\RR_+)$, but we will construct our quantum algorithms by discretizing~\cref{eqn:lchs_matrixhardy} and applying the quantum LCU subroutine, which requires $g \in L^1(\RR_+)$. Nevertheless, there are still some cases where $g(t)$ satisfies~\cref{eqn:lchs_matrixhardy}, is not in $L^1(\RR_+)$, and allows an efficient quantum implementation. 
An example is 
\begin{equation}
    g(t) = \sum_{j} \delta_{t_j}(t) + \widetilde{g}(t), 
\end{equation}
where $\delta_{t_j}(t)$ is the Dirac delta function at $t_j \geq 0$ and $\widetilde{g}(t)$ is in $L^1(\RR_+)$. 
In this case, by~\cref{eqn:lchs_matrixhardy}, the corresponding matrix function $h(A)$ can be written as 
\begin{equation}
    h(A) = \sum_j \int_{\mathbb{R}} \frac{f(k)}{ 1-ik} e^{-i t_j (kL+H)} \ud k + \int_0^{\infty} \int_{\mathbb{R}} \frac{f(k)\widetilde{g}(t)}{ 1-ik} e^{-it (kL+H)} \ud k \ud t. 
\end{equation}
Each integral in the summation is just the matrix exponential $e^{-A t_j}$ according to~\cref{thm:LCHS_improved} and can be efficiently implemented by the standard LCHS algorithm~\cite{AnChildsLin2023}. 
However, for technical simplicity, in this work we mainly focus on the case where the inverse Laplace transform is assumed to be in $L^1(\RR_+)$. 

Based on~\cref{thm:LCHS_hardy}, we may design an algorithm that implements a block encoding of $h(A)$ and prepare a quantum state proportional to $h(A)\ket{\psi}$. 
We first discretize the integral~\cref{eqn:lchs_matrixhardy} to approximate $h(A)$ by a discrete linear combination of multiple Hamiltonian simulation problems. 
The algorithm then implements each Hamiltonian simulation time-evolution operator by QSVT for optimal asymptotic scaling and combines the simulations using the LCU method.

\subsection{Discretization of the integral}\label{sec:Lap-LCHS_quadrature}

We first truncate~\cref{eqn:lchs_matrixhardy} into a finite domain $[-K,K]\times [0,T]$, and then use the simplest Riemann sum to approximate the integral. 
Specifically, let 
\begin{equation}
    U(k,t) = e^{-it(kL+H)}. 
\end{equation}
We use the approximations
\begin{align}
    h(A) &= \int_0^{\infty} \int_{\mathbb{R}} \frac{f(k){g}(t)}{ 1-ik} U(k,t) \ud k \ud t \\
    & \approx \int_0^{T} \int_{-K}^K \frac{f(k){g}(t)}{ 1-ik} U(k,t) \ud k \ud t \\
    & \approx \sum_{j=0}^{M_k-1}\sum_{l=0}^{M_t-1} c_j \hat{c}_l U(k_j,t_l). \label{eqn:quadrature}
\end{align}
Here $M_k$ and $M_t$ are the numbers of grid points for discretizing the variables $k$ and $t$, respectively. 
The corresponding step sizes are $h_k = 2K/M_k$ and $h_t = T/M_t$. 
The nodes are chosen to be $k_j = -K+jh_j$ and $t_l = lh_t$, and the coefficients are $c_j = h_k \frac{f(k_j)}{1-ik_j}$ and $\hat{c}_l = h_t {g}(t_l)$. 
We bound the approximation errors in the following result, whose proof can be found in~\cref{app:quadrature_error}. 

\begin{lem}\label{lem:quadrature}
    Suppose that $g(t) \in C^1(\RR_+) \cap L^1(\RR_+)$. Then
    \begin{enumerate}
        \item The truncation error can be bounded as 
        \begin{align}
            &\norm{ h(A) - \int_0^{T} \int_{-K}^K \frac{f(k){g}(t)}{ 1-ik} U(k,t) \ud k \ud t } \nonumber\\
            &\quad \leq \norm{f}_{L^1(\mathbb{R})} \norm{{g}}_{L^1((T,\infty))} + \norm{f}_{L^1(\mathbb{R}\setminus[-K.K])} \norm{{g}}_{L^1(\RR_+)}. 
        \end{align}
        \item The quadrature error can be bounded as  
        \begin{align}
            & \norm{ \int_0^{T} \int_{-K}^K \frac{f(k){g}(t)}{ 1-ik} U(k,t) \ud k \ud t - \sum_{j=0}^{M_k-1}\sum_{l=0}^{M_t-1} c_j \hat{c}_l U(k_j,t_l) } \nonumber\\
            &\quad \leq 2 K T h_k  \left( \norm{f}_{L^{\infty}(\mathbb{R})} \norm{t {g}}_{L^{\infty}(\RR_+)} \norm{A} + \left( \norm{f'}_{L^{\infty}(\mathbb{R})} + \norm{f}_{L^{\infty}(\mathbb{R})} \right) \norm{{g}}_{L^{\infty}(\RR_+)} \right) \nonumber\\
            &\qquad + 2 K T h_t  \left( 2\norm{f}_{L^{\infty}(\mathbb{R})} \norm{{g}}_{L^{\infty}(\RR_+)} \norm{A} + \norm{f}_{L^{\infty}(\mathbb{R})} \norm{{g}'}_{L^{\infty}(\RR_+)} \right). 
        \end{align}
        \item In order to bound the overall approximation error $\norm{h(A) - \sum_{j,l}c_j \hat{c}_l U(k_j,t_l)}$ by $\epsilon$, it suffices to choose $K$ and $T$ such that $\norm{{g}}_{L^1((T,\infty))} = \Or(\epsilon/\norm{f}_{L^1(\mathbb{R})})$ and $\norm{f}_{L^1(\mathbb{R}\setminus[-K.K])}  = \Or(\epsilon/\norm{{g}}_{L^1(\RR_+)})$, and the numbers of grid points to be 
        \begin{equation}
            M_k = \Or\left( \frac{K^2T}{\epsilon} \left( \norm{f}_{L^{\infty}(\mathbb{R})} \norm{t {g}}_{L^{\infty}(\RR_+)} \norm{A} + \left( \norm{f'}_{L^{\infty}(\mathbb{R})} + \norm{f}_{L^{\infty}(\mathbb{R})} \right) \norm{{g}}_{L^{\infty}(\RR_+)} \right) \right)
        \end{equation}
        and 
        \begin{equation}
            M_t = \Or\left( \frac{KT^2}{\epsilon}\left( \norm{f}_{L^{\infty}(\mathbb{R})} \norm{{g}}_{L^{\infty}(\RR_+)} \norm{A} + \norm{f}_{L^{\infty}(\mathbb{R})} \norm{{g}'}_{L^{\infty}(\RR_+)} \right)  \right).  
        \end{equation}
    \end{enumerate}
\end{lem}

\cref{lem:quadrature} shows that the quadrature error and the numbers of grid points depend polynomially on the truncation parameters $K$ and $T$ such that the tail integrals of $f(k)$ and $g(t)$, respectively, are sufficiently small. If we use the asymptotically near-optimal kernel function $f(k)$ defined in~\cref{eqn:kernel_intro}, then it suffices to choose $K = \Or\left( \left( \log( \norm{g}_{L^{1}(\RR_+)}/\epsilon ) \right)^{1/\beta} \right)$ for any $\beta \in (0,1)$, which scales only poly-logarithmically in the inverse error. The scaling of $T$, however, depends on the decay rate of the inverse Laplace transform $g(t)$ as $t \rightarrow \infty$ and should be carefully analyzed on a case-by-case basis. For example, in computing the power of matrix inverse (see~\cref{sec:power_inverse}) and solving linear differential equations with mass matrices (see~\cref{sec:firstorder_mass}), the inverse Laplace transforms both decay exponentially in $t$, resulting a poly-logarithmic scaling of $T$ with respect to $1/\epsilon$. 

\subsection{Oracles}\label{sec:Lap-LCHS_oracles}

For a matrix $A$, we typically assume access to a block encoding $O_A$ such that 
\begin{equation}
    \left(\bra{0}\otimes I\right) O_A \left(\ket{0}\otimes I\right) = \frac{A}{\alpha_A}. 
\end{equation}
Here $\alpha_A \geq \|A\|$ is the block encoding factor. 
We may alternatively assume access to block encodings $O_L$ and $O_H$ of the matrix real and imaginary parts of $A$, such that 
\begin{equation}
    \left(\bra{0}\otimes I\right) O_L \left(\ket{0}\otimes I\right) = \frac{L}{\alpha_L}, \quad \left(\bra{0}\otimes I\right) O_H \left(\ket{0}\otimes I\right) = \frac{H}{\alpha_H}. 
\end{equation}
Notice that if we have $O_A$, then according to the equations $L = (A+A^{\dagger})/2$ and $H = (A-A^{\dagger})/(2i)$, we may apply LCU to construct $O_L$ and $O_H$ with $\alpha_L = \alpha_H = \alpha_A$, using one additional ancilla qubit and two queries to $O_A$. 
On the other hand, a block encoding of $A$ can also be constructed from block encodings of $L$ and $H$ by LCU with $\alpha_A = \alpha_L+\alpha_H$ and one query to each of $O_L$ and $O_H$. 

If we are interested in approximating the state $h(A)\ket{\psi}/\norm{h(A)\ket{\psi}}$, we assume a state preparation oracle $O_{\psi}$ for $\ket{\psi}$ such that $O_{\psi}\ket{0} = \ket{\psi}$. 

In our algorithm, we use a few more unitaries to encode information about the quadrature formula~\cref{eqn:quadrature} into quantum states. These unitaries can be constructed with much lower cost than the block-encoding $O_A$ of the matrix $A$ and the state preparation oracle $O_{\psi}$ of the input state, as they are independent of the dimension of the eigenvalue transformation problem. We refer to them as \emph{quadrature unitaries} and discuss their definitions and constructions here. 

For the nodes $k_j$ and $t_l$, we can construct unitaries $O_k$ and $O_t$ to encode them in binary, i.e., 
\begin{equation}
    O_k \ket{j}\ket{0} = \ket{j} \ket{k_j}, \quad O_t \ket{l}\ket{0} = \ket{l} \ket{t_l}. 
\end{equation}
For the coefficients, we construct a pair of state preparation unitaries. 
Specifically, for $c = (c_0,\cdots,c_{M_k-1})$, we assume $(O_{c,l}, O_{c,r})$ such that 
\begin{align}
    O_{c,l} \ket{0} = \frac{1}{\sqrt{\norm{c}_1}} \sum_{j=0}^{M_k-1} \overline{\sqrt{c_j}} \ket{j}, \quad O_{c,r} \ket{0} = \frac{1}{\sqrt{\norm{c}_1}} \sum_{j=0}^{M_k-1} \sqrt{c_j} \ket{j}. 
\end{align}
Here $\overline{z}$ denotes the complex conjugate of $z \in \CC$, and $\sqrt{\cdot}$ denotes the principal branch of the square root. 
Similarly, for $\hat{c} = (\hat{c}_0,\cdots,\hat{c}_{M_t-1})$, we construct $(O_{\hat{c},l}, O_{\hat{c},r})$ such that 
\begin{align}
    O_{\hat{c},l} \ket{0} = \frac{1}{\sqrt{\norm{\hat{c}}_1}} \sum_{l=0}^{M_t-1} \overline{\sqrt{\hat{c}_l}} \ket{l}, \quad O_{\hat{c},r} \ket{0} = \frac{1}{\sqrt{\norm{\hat{c}}_1}} \sum_{l=0}^{M_t-1} \sqrt{\hat{c}_l} \ket{l}. 
\end{align}
Notice that the unitaries $O_{c,l}$, $O_{c,r}$, $O_{\hat{c},l}$, and $O_{\hat{c},r}$ prepare superpositions of $M$ basis states, where $M$ can be polynomial in terms of $K$, $T$, and $\epsilon$ according to~\cref{lem:quadrature}. 
In general, we can prepare an $M$-dimensional quantum state with cost $\Or(M)$ \cite{SBM06}, but this might incur a gate complexity that is polynomial in the inverse error $1/\epsilon$, as $M$ scales polynomially in $1/\epsilon$ for low-order quadrature formulas (as shown in~\cref{lem:quadrature}). However, since the amplitudes of the states are known integrable functions evaluated at discrete points, the state preparation circuits can be constructed more efficiently, in time only $\Or(\poly\log M)$~\cite{GroverRudolph2002,McardleGilyenBerta2022} (see~\cref{app:prepare_oracle} for a detailed discussion). 
Alternatively, we may use a high-order quadrature rule for integral discretization following~\cite{AnChildsLin2023}, such as composite Gaussian quadrature, to improve the scaling of $M$ on $K$, $T$, and $\epsilon$, at the expense of introducing high-order derivative dependence.

\subsection{Algorithm}\label{sec:Lap-LCHS_algorthm}

The basic approach to implementing $h(A)$ is to use LCU to linearly combine Hamiltonian simulation operators $U(k_j,t_l)$. 
We first describe the construction of a select oracle for $U(k,t)$, a prerequisite for LCU.
We can write 
\begin{equation}
    U(k,t) = e^{-it(kL+H)} = e^{ -i T (K\alpha_L + \alpha_H) \widetilde{H} }, \quad \widetilde{H} = \frac{t}{T} \frac{1}{K\alpha_L + \alpha_H} (kL+H). 
\end{equation}
Thus it suffices to construct a coherent block encoding of $\widetilde{H}$ and use it as the input model for the QSVT circuit for the time-evolution operator $e^{ -i T (K\alpha_L + \alpha_H) \widetilde{H} }$. 

Let us start with the state 
\begin{equation}
    \ket{j}\ket{l} \ket{0}_j \ket{0}_l \ket{0}_{R} \ket{0}_{R_k} \ket{0}_{R_t} \ket{0}_a \ket{\psi}. 
\end{equation}
Here $\ket{\psi}$ is the input state in the system register, and
$\ket{j}$ 
and $\ket{l}$ are indices for the $k$ and $t$ variables, respectively. 
We introduce six additional ancilla registers.
Specifically, $\ket{0}_j$ and $\ket{0}_l$ encode in binary the nodes $k_j$ and $t_l$, respectively; $\ket{0}_R$, $\ket{0}_{R_k}$, and $\ket{0}_{R_t}$ are single qubits for rotations; and $\ket{0}_a$ represents the ancilla register of the block encoding. 
We first apply $O_k$ and $O_t$ to compute the nodes, giving
\begin{equation}\label{eqn:alg_sel_H_in}
    \ket{j}\ket{l} \ket{k_j}_j \ket{t_l}_l \ket{0}_R \ket{0}_{R_k} \ket{0}_{R_t} \ket{0}_a \ket{\psi}. 
\end{equation}
Now, we construct a controlled block encoding of $kL+H$ by LCU. 
Applying a controlled rotation $\text{c-R}$ acting as
\begin{equation}
    \text{c-R}: \ket{k}\ket{0} \rightarrow \ket{k} \left( \frac{\sqrt{\alpha_L k}}{\sqrt{\alpha_L |k|+\alpha_H}}\ket{0} + \frac{\sqrt{\alpha_H}}{\sqrt{\alpha_L |k|+\alpha_H}}\ket{1}\right), 
\end{equation}
we obtain 
\begin{equation}
    \ket{j}\ket{l} \ket{k_j}_j \ket{t_l}_l  \left( \frac{\sqrt{\alpha_L k_j}}{\sqrt{\alpha_L |k_j|+\alpha_H}}\ket{0}_R + \frac{\sqrt{\alpha_H}}{\sqrt{\alpha_L |k_j|+\alpha_H}}\ket{1}_R \right) \ket{0}_{R_k} \ket{0}_{R_t} \ket{0}_a \ket{\psi}. 
\end{equation}
Applying the controlled block encodings $\ket{0}_R\bra{0}_R \otimes U_L$ and $\ket{1}_R\bra{1}_R \otimes U_H$ yields 
\begin{align}
    &\quad \ket{j}\ket{l} \ket{k_j}_j \ket{t_l}_l   \frac{\sqrt{\alpha_L k_j}}{\sqrt{\alpha_L |k_j|+\alpha_H}}\ket{0}_R \ket{0}_{R_k} \ket{0}_{R_t} \ket{0}_a \frac{L}{\alpha_L}\ket{\psi} \\
    & + \ket{j}\ket{l} \ket{k_j}_j \ket{t_l}_l   \frac{\sqrt{\alpha_H}}{\sqrt{\alpha_L |k_j|+\alpha_H}}\ket{1}_R \ket{0}_{R_k} \ket{0}_{R_t} \ket{0}_a \frac{H}{\alpha_H}\ket{\psi} + \ket{\perp_a}. 
\end{align}
Here, $\ket{\perp_a}$ is a possibly unnormalized state orthogonal to $\ket{0}_a$ in the ancilla register labeled by $a$. Later on, we will append more subscripts to indicate that it is orthogonal in more ancilla registers. 
Applying $\text{c-}\overline{\text{R}}^{\dagger}$ where $\text{c-}\overline{\text{R}}$ is the controlled rotation gate obtained by taking the conjugates of the coefficients in $\text{c-R}$, we obtain 
\begin{equation}
    \ket{j}\ket{l} \ket{k_j}_j \ket{t_l}_l   \ket{0}_R \ket{0}_{R_k} \ket{0}_{R_t} \ket{0}_a \frac{k_j L+H}{\alpha_L |k_j|+\alpha_H}\ket{\psi} + \ket{\perp_{R,a}}. 
\end{equation}
We further shrink the rescaling factor of $k_j L+H$ by applying two additional controlled rotations on $R_k$ and $R_k$, acting as 
\begin{equation}
    \ket{k}\ket{0} \rightarrow  \ket{k} \left( \frac{\alpha_L |k| + \alpha_H}{\alpha_L K + \alpha_H} \ket{0} + \sqrt{1-\left|\frac{\alpha_L |k| + \alpha_H}{\alpha_L K + \alpha_H}\right|^2} \ket{1} \right), 
\end{equation}
and 
\begin{equation}
    \ket{t}\ket{0} \rightarrow \ket{t} \left( \frac{t}{T}\ket{0} + \sqrt{1-\left|\frac{t}{T}\right|^2}\ket{1}\right),
\end{equation}
respectively.
Then we obtain 
\begin{equation}
    \ket{j}\ket{l} \ket{k_j}_j \ket{t_l}_l   \ket{0}_R \ket{0}_{R_k} \ket{0}_{R_t} \ket{0}_a \frac{t_l (k_j L+H) }{T(\alpha_L K +\alpha_H)}\ket{\psi} + \ket{\perp_{R,R_k,R_t,a}}. 
\end{equation}
Uncomputing the $j$ and $l$ registers by applying $O_k^{\dagger}$ and $O_t^{\dagger}$ gives 
\begin{equation}\label{eqn:alg_sel_H_out}
    \ket{j}\ket{l} \ket{0}_j \ket{0}_l   \ket{0}_R \ket{0}_{R_k} \ket{0}_{R_t} \ket{0}_a \frac{t_l (k_j L+H) }{T(\alpha_L K +\alpha_H)}\ket{\psi} + \ket{\perp_{R,R_k,R_t,a}}. 
\end{equation}
We have shown how to apply a sequence of operators to map~\cref{eqn:alg_sel_H_in} to~\cref{eqn:alg_sel_H_out}. 
By the definition of block encoding, this sequence of operators, denoted by $U_{t(kL+H)}$, is indeed a controlled block encoding of $t(kL+H)$, satisfying
\begin{equation}
    \left(\bra{0}_{a'} \otimes I\right) U_{t(kL+H)} \left(\ket{0}_{a'} \otimes I\right) = \sum_{j=0}^{M_k-1} \sum_{l=0}^{M_t-1} \ket{j}\bra{j} \otimes \ket{l}\bra{l} \otimes \frac{t_l (k_j L+H)}{T(\alpha_L K + \alpha_H)}. 
\end{equation}
Here, the ancilla register labeled by $a'$ combines all the previous ancilla registers. 

Now we use $U_{t(kL+H)}$ as the block encoding in the QSVT circuit for the time-evolution operator $e^{-iT(\alpha_L K + \alpha_H)\widetilde{H}}$. 
Then, by~\cite[Corollary 60]{GilyenSuLowEtAl2019}, we obtain the select oracle 
\begin{equation}\label{eqn:alg_select_oracle}
    \text{SEL} = \sum_{j=0}^{M_k-1} \sum_{l=0}^{M_t-1} \ket{j}\bra{j} \otimes \ket{l}\bra{l} \otimes W_{j,l}. 
\end{equation}
Here $W_{j,l}$ block encodes another matrix $V_{j,l} \approx U(k_j,t_l)$. 
Finally, by the LCU lemma (see, e.g.,~\cite[Lemma 22]{AnChildsLin2023}), the operator 
\begin{equation}
    (O_{c,l}^{\dagger} \otimes O_{\hat{c},l}^{\dagger} \otimes I) \text{SEL} (O_{c,r}\otimes O_{\hat{c},r} \otimes I)
\end{equation} 
gives a block encoding of $\frac{1}{\norm{c}_1\norm{\hat{c}}_1} \sum_{j,l} c_j\hat{c}_l V_{j,l} $, which is an approximation of $\frac{1}{\norm{c}_1\norm{\hat{c}}_1} h(A)$.

\subsection{Complexity analysis}\label{sec:Lap-LCHS_complexity}

\begin{thm}\label{thm:complexity_block_encoding}
    Let $f, {g}, h$ be functions satisfying the assumptions in~\cref{thm:LCHS_hardy}, and let $A$ be a matrix with positive semi-definite matrix real part. 
    Suppose that we are given the oracles in~\cref{sec:Lap-LCHS_oracles}. 
    Then Lap-LCHS can implement an $(\alpha,\epsilon)$-block encoding of $h(A)$ with the following properties.
    \begin{enumerate}
        \item The block encoding factor satisfies 
        \begin{equation}
            \alpha = \Or(\norm{f}_{L^1(\mathbb{R})}\norm{{g}}_{L^1(\RR_+)}). 
        \end{equation}
        \item It suffices to choose $K$ and $T$ such that $\norm{{g}}_{L^1((T,\infty))} = \Or(\epsilon/\norm{f}_{L^1(\mathbb{R})})$ and $\norm{f}_{L^1(\mathbb{R}\setminus[-K,K])}  = \Or(\epsilon/\norm{{g}}_{L^1(\RR_+)})$, and choose $M_k,M_t$ as in~\cref{lem:quadrature}. 
        \item The algorithm uses 
        \begin{equation}
            \Or\left( \alpha_A K T + \log\left(\frac{\norm{f}_{L^1(\mathbb{R})}\norm{{g}}_{L^1(\RR_+)}}{\epsilon}\right)\right)
        \end{equation}
        queries to $O_A$, and $\Or(1)$ queries to the oracles $O_k,O_t,O_c,O_{\hat{c}}$. 
    \end{enumerate}
\end{thm}
\begin{proof}
    Following the steps in~\cref{sec:Lap-LCHS_algorthm}, we can construct a controlled $(T(\alpha_L K + \alpha_H),0)$-block encoding of $t(kL+H)$, using $\Or(1)$ queries to $O_A$. 
    According to~\cite[Corollary 60]{GilyenSuLowEtAl2019}, for $\epsilon_1 > 0$ to be determined later, we can implement 
    \begin{equation}
    \text{SEL} = \sum_{j=0}^{M_k-1} \sum_{l=0}^{M_t-1} \ket{j}\bra{j} \otimes \ket{l}\bra{l} \otimes W_{j,l}, 
    \end{equation}
    where $W_{j,l}$ is a $(1,0)$-block encoding of $V_{j,l}$ such that $\norm{V_{j,l}-U(k_j,t_l)} \leq \epsilon_1$, using 
    \begin{equation}
        \Or\left( T(\alpha_L K + \alpha_H) + \log\left(\frac{1}{\epsilon_1}\right) \right) = \Or\left( \alpha_A K T + \log\left(\frac{1}{\epsilon_1}\right)\right)
    \end{equation}
    queries to $O_A$. 
    Then the LCU lemma (see, e.g.,~\cite[Lemma 22]{AnChildsLin2023}) ensures that $(O_{c,l}^{\dagger} \otimes O_{\hat{c},l}^{\dagger} \otimes I) \text{SEL} (O_{c,r}\otimes O_{\hat{c},r} \otimes I)$ gives a $(\norm{c}_1\norm{\hat{c}}_1,0)$-block encoding of $ \sum_{j,l} c_j \hat{c}_l V_{j,l}$. 
    By the definition of the coefficients, we can further bound the block encoding factor as $\norm{c}_1\norm{\hat{c}}_1 = \Or(\norm{f}_{L^1(\mathbb{R})}\norm{{g}}_{L^1(0,\infty)})$. 

    We now estimate the approximation error between $\sum_{j,l} c_j \hat{c}_l V_{j,l}$ and the ideal operator $h(A)$ and determine the choices of the parameters. 
    By the triangle inequality, we have 
    \begin{align}
        \norm{h(A) - \sum_{j,l} c_j \hat{c}_l V_{j,l}} & \leq \norm{ \sum_{j,l} c_j \hat{c}_l V_{j,l} - \sum_{j,l} c_j \hat{c}_l U(k_j,t_l) } + \norm{h(A) - \sum_{j,l} c_j \hat{c}_l U(k_j,t_l)} \\
        & \leq  \sum_{j,l} |c_j| |\hat{c}_l| \norm{V_{j,l}-U(k_j,t_l)} + \norm{h(A) - \sum_{j,l} c_j \hat{c}_l U(k_j,t_l)} \\
        & \leq \norm{c}_1 \norm{\hat{c}}_1 \epsilon_1 + \norm{h(A) - \sum_{j,l} c_j \hat{c}_l U(k_j,t_l)}. 
    \end{align}
    To bound this by $\epsilon$, it suffices to choose 
    \begin{equation}
        \epsilon_1 = \Or\left( \frac{\epsilon}{\norm{f}_{L^1(\mathbb{R})}\norm{{g}}_{L^1(\RR_+)}} \right) 
    \end{equation}
    and bound the quadrature error $\norm{h(A) - \sum_{j,l} c_j \hat{c}_l U(k_j,t_l)}$ by $\Or(\epsilon)$. 
    \cref{lem:quadrature} gives sufficient choices of $K,T$ and $M_k,M_t$. 
\end{proof}

The state $h(A)\ket{\psi}/\norm{h(A)\ket{\psi}}$ can then be approximated by applying the block encoding of $h(A)$ onto the input state $\ket{\psi}$ and boosting the success probability by amplitude amplification. 

\begin{cor}\label{cor:complexity_state}
    Let $f, {g}, h$ be functions satisfying the assumptions in~\cref{thm:LCHS_hardy}, let $A$ be a matrix with positive semi-definite real part, and let $\ket{\psi}$ be an input state.
    Suppose that we are given the oracles described in~\cref{sec:Lap-LCHS_oracles}. 
    Then Lap-LCHS can prepare an $\epsilon$-approximation of the state $h(A)\ket{\psi}/\norm{h(A)\ket{\psi}}$ with $\Omega(1)$ success probability and a flag indicating success, using 
    \begin{equation}
            \Or\left( \frac{\norm{f}_{L^1(\mathbb{R})}\norm{{g}}_{L^1(\RR_+)}}{\norm{h(A) \ket{\psi}}} \left(\alpha_A K T + \log\left(\frac{\norm{f}_{L^1(\mathbb{R})}\norm{{g}}_{L^1(\RR_+)}}{\norm{h(A)\ket{\psi}} \epsilon}\right)\right) \right)
    \end{equation}
    queries to $O_A$, and 
    \begin{equation}
        \Or\left( \frac{\norm{f}_{L^1(\mathbb{R})}\norm{{g}}_{L^1(\RR_+)}}{\norm{h(A) \ket{\psi}}} \right)
    \end{equation}
    queries to $O_{\psi}$ and quadrature unitaries. 
    Here $K$ and $T$ are truncation parameters such that $\norm{{g}}_{L^1((T,\infty))} = \Or(\epsilon \norm{h(A)\ket{\psi}}/\norm{f}_{L^1(\mathbb{R})})$ and $\norm{f}_{L^1(\mathbb{R}\setminus[-K,K])}  = \Or(\epsilon \norm{h(A)\ket{\psi}}/\norm{{g}}_{L^1(\RR_+)})$. 
\end{cor}
\begin{proof}
    We first use~\cref{thm:complexity_block_encoding} to construct an $(\alpha,\epsilon')$-block encoding of $h(A)$, where the precision parameter $\epsilon'$ will be determined later. 
    Applying this block encoding to $\ket{0}\ket{\psi}$ gives 
    \begin{equation}
        \frac{1}{\alpha}\ket{0} B \ket{\psi} + \ket{\perp}, 
    \end{equation}
    where $B$ is an operator such that $\norm{B - h(A)} \leq \epsilon'$. 
    Upon projecting the ancilla onto $\ket{0}$, we obtain the state $B\ket{\psi}/\norm{B\ket{\psi}}$. 
    Furthermore, 
    \begin{equation}
        \norm{\frac{B\ket{\psi}}{\norm{B\ket{\psi}}} - \frac{h(A)\ket{\psi}}{\norm{h(A)\ket{\psi}}} } \leq \frac{2}{\norm{h(A)\ket{\psi}}} \norm{ B\ket{\psi} - h(A)\ket{\psi} } \leq \frac{2\epsilon'}{\norm{h(A)\ket{\psi}}}. 
    \end{equation}
    To bound this by $\epsilon$, it suffices to choose $\epsilon' = \norm{h(A)\ket{\psi}} \epsilon/2$. 
    According to~\cref{thm:complexity_block_encoding}, in each run of the algorithm, we use 
    \begin{equation}
            \Or\left( \alpha_A K T + \log\left(\frac{\norm{f}_{L^1(\mathbb{R})}\norm{{g}}_{L^1(\RR_+)}}{\norm{h(A)\ket{\psi}} \epsilon}\right)\right)
    \end{equation}
    queries to $O_A$, and $\Or(1)$ queries to the state preparation oracle $O_{\psi}$ and quadrature unitaries.
    
    With amplitude amplification, the number of repetitions for constant success probability is 
    \begin{equation}
        \Or\left( \frac{\alpha}{\norm{B\ket{\psi}}} \right) = \Or\left( \frac{\norm{f}_{L^1(\mathbb{R})}\norm{{g}}_{L^1(\RR_+)}}{\norm{h(A) \ket{\psi}}} \right),  
    \end{equation}
    where we have used $\norm{B\ket{\psi}} \geq \norm{h(A)\ket{\psi}} - \epsilon' = (1-\epsilon/2)\norm{h(A)\ket{\psi}}$. 
    This contributes to another multiplicative factor in the claimed complexity and completes the proof. 
\end{proof}

\section{Applications}\label{sec:applications}

Here we discuss several applications of eigenvalue transformation problems and analyze the query complexity of the Lap-LCHS algorithm for them. 
These applications include linear inhomogeneous differential equations, powers of the matrix inverse, linear differential equations with non-normal mass matrices, and second-order differential equations. 
Throughout this section, in the LCHS formula we always use the asymptotically optimal kernel function $f(k)$ defined in~\cref{eqn:kernel_intro}. We regard the parameter $\beta \in (0,1)$ in~\cref{eqn:kernel_intro} as fixed; the big-$\Or$ constants in this section may depend on $\beta$. 

\subsection{Linear inhomogeneous differential equations}\label{sec:linear_inhom}

Consider the linear inhomogeneous differential equation in \cref{eqn:firstorder_inhom}. By variation of constants, the solution is
\begin{equation}\label{eqn:inhom_lcus}
u(T)=\int_0^T e^{-A(t-s)}\ket{\psi}\ud s= h(A)\ket{\psi},
\end{equation}
where
\begin{equation}
h(z)=\frac{1}{ z}(1-e^{-T z}). 
\end{equation}
This function has the inverse Laplace transform ${g}(t)={\bm 1}_{[0,T]}(t)$, which is the indicator function. Lap-LCHS for evaluating $h(A)\ket{\psi}$ uses the integral representation in \cref{eqn:inhom_lcus}. 
The overall complexity can be estimated as a direct consequence of~\cref{cor:complexity_state}, by noticing that $\norm{g}_{L^1(\RR_+)} = T$ and the truncation parameter $K = \Or((\log(T/( \norm{u(T)}\epsilon)))^{1/\beta})$ for fixed $\beta \in (0,1)$. 

\begin{cor}\label{cor:inhomogeneous_DE}
    Consider the linear inhomogeneous differential equation in~\cref{eqn:firstorder_inhom}, where the coefficient matrix $A$ has a positive semi-definite matrix real part. 
    Then, for any fixed $\beta \in (0,1)$, Lap-LCHS can prepare an $\epsilon$-approximation of the state $u(T)/\norm{u(T)}$ with $\Omega(1)$ success probability and a flag indicating success, using 
    \begin{equation}
        \Or\left( \frac{ 1 }{\norm{u(T)}} \alpha_A T^2 \left( \log\left( \frac{T}{\norm{u(T)} \epsilon} \right) \right)^{1/\beta}  \right)
    \end{equation}
    queries to the block encoding of $A$ with block encoding factor $\alpha_A$, and 
    \begin{equation}
        \Or\left( \frac{ T }{\norm{ u(T) }} \right)
    \end{equation}
    queries to the state preparation oracle for $\ket{\psi}$. 
    Here $\beta \in (0,1)$ is the order parameter in the kernel function of the LCHS formula. 
\end{cor}

The best known quantum algorithm prior to LCHS was the truncated Taylor series method proposed in~\cite{BerryChildsOstranderEtAl2017}. 
This method discretizes the ODE in~\cref{eqn:firstorder_inhom} using the truncated Taylor series, formulates the discretized evolution as a dilated linear system of equations, and then solves it with the LCU-based quantum linear system algorithm~\cite{ChildsKothariSomma2017}. 
Reference~\cite{BerryChildsOstranderEtAl2017} gives a detailed complexity analysis using a sparse input oracle for the coefficient matrix. 
A recent work~\cite{BerryCosta2022}, besides generalizing this idea to the time-dependent case, establishes an improved complexity estimate with block encoding oracle access and an optimal quantum linear system algorithm~\cite{CostaAnSandersEtAl2022}. 

A comparison between Lap-LCHS and the truncated Taylor series method is given in~\cref{tab:comparison_app_1st_ODE_inhomo}. We see that Lap-LCHS has better state preparation cost by fully eliminating the explicit dependence on $\alpha_A$ and $\epsilon$. 
In terms of the matrix query complexity, Lap-LCHS is better in terms of $\epsilon$ as it has a lower degree of the $\log(1/\epsilon)$ term. 
However, for simulation up to fixed accuracy, by noticing that $\max_t\norm{u(t)} \leq T$, Lap-LCHS only has at most comparable (and sometimes worse) matrix query complexity compared to the truncated Taylor series method. 

\begin{table}[t]
    \renewcommand{\arraystretch}{2}
    \centering
    \scalebox{1}{
    \begin{tabular}{c|c|c}\hline\hline
        \textbf{Method} & 
        \textbf{Queries to $A$} & 
        \textbf{Queries to $\ket{\psi}$} \\\hline 
        This work (\cref{cor:inhomogeneous_DE}) & 
        $\wt{\Or}\left( \frac{ 1 }{\norm{u(T)}} \alpha_A T^2 \left( \log\left( \frac{1}{\epsilon} \right) \right)^{1/\beta} \right)$  & 
        $\Or\left( \frac{ T }{\norm{ u(T) }} \right)$ \\\hline
        Taylor~\cite{BerryChildsOstranderEtAl2017,BerryCosta2022} & 
        $\wt{\Or}\left( \frac{ \max_t\norm{u(t)} }{\norm{u(T)}} \alpha_A T \left( \log\left( \frac{1}{\epsilon} \right) \right)^{2} \right)$  & 
        $ \wt{\Or}\left( \frac{ \max_t\norm{u(t)} }{\norm{u(T)}} \alpha_A T \log\left( \frac{1}{\epsilon} \right)  \right) $ \\\hline\hline
    \end{tabular}
    }
    \caption{ Comparison between Lap-LCHS and the previous approach for solving the ODE~\cref{eqn:firstorder_inhom}. 
    Here, we assume the Cartesian decomposition $A = L+iH$ with $L \succeq 0$. $\alpha_A$ is the block encoding factor of $A$, $T$ is the evolution time, and $\epsilon$ is the tolerated error in the output state. }
    \label{tab:comparison_app_1st_ODE_inhomo}
\end{table}

\subsection{Power of matrix inverse}\label{sec:power_inverse}

Consider the problem of solving the shifted linear system of equations 
\begin{equation}\label{eqn:shifted_LSE}
    (\eta I + A) x = \ket{b}. 
\end{equation}
Here $\eta > 0$ is the shifting parameter. 
The goal of a quantum algorithm for~\cref{eqn:shifted_LSE} is to approximately prepare the quantum state $(\eta I + A)^{-1} \ket{b} / \norm{(\eta I + A)^{-1} \ket{b}}$. 
Here, we discuss a more general problem of preparing the state 
\begin{equation}
    \ket{x} = \frac{ (\eta I + A)^{-p} \ket{b} }{ \norm{(\eta I + A)^{-p} \ket{b}} }, 
\end{equation}
where $p > 0$ is a real positive parameter.  

In the Lap-LCHS framework, we choose the function 
\begin{equation}
    h(z) = (\eta + z)^{-p}. 
\end{equation}
Its inverse Laplace transform is 
\begin{equation}
    g(t) = \frac{1}{\Gamma(p)} t^{p-1} e^{-\eta t}, 
\end{equation}
where $\Gamma(p)$ denotes the gamma function. 
The following result shows the complexity of our algorithm. Its proof can be found in~\cref{app:applications_proof_linear_system}. 

\begin{cor}\label{cor:app_linear_system}
    Let $A$ be a matrix with positive semi-definite real part, and let $\eta > 0$ and $p > 0$. 
    Then for any fixed $\beta \in (0,1)$, Lap-LCHS can prepare an $\epsilon$-approximation of the state $\ket{x} = x/\norm{x}$ where $x = (\eta I + A)^{-p}\ket{b}$, with $\Omega(1)$ success probability and a flag indicating success, using 
    \begin{equation}
            \Or\left( \frac{\alpha_A}{\eta^{p+1} \norm{x}} \left(\log\left( \frac{1}{\epsilon \eta \norm{x}} \right)\right)^{1+1/\beta}  \right)
    \end{equation}
    queries to the block encoding of $A$ with block encoding factor $\alpha_A$, and
    \begin{equation}
        \Or\left( \frac{1}{\eta^p \norm{x} } \right)
    \end{equation}
    queries to the state preparation oracle for $\ket{b}$. 
\end{cor}

\subsection{Linear differential equations with mass matrices}\label{sec:firstorder_mass}

Consider the differential equation
\begin{equation}\label{eqn:app_1st_ODE}
A\frac{\ud u}{\ud t}=-u
\end{equation}
where $A = L+iH$ for Hermitian matrices $L,H$, and we assume the matrix real part satisfies $L \succeq \gamma > 0$. 
The goal of a quantum algorithm for~\cref{eqn:app_1st_ODE} is to prepare a quantum state approximating the normalized solution $u(T)/\|u(T)\|$ where $u(T) = e^{-T A^{-1}}u(0)$. 
We consider two types of initial conditions, namely 
\begin{equation}\label{eqn:app_1st_ini_cond_w_A}
    u(0) = A^{-1} u_0 
\end{equation}
or 
\begin{equation}\label{eqn:app_1st_ini_cond_w/o_A}
    u(0) = u_0. 
\end{equation}
In both cases, we are given an oracle for preparing the state $\ket{u_0} = u_0/\norm{u_0}$.

\subsubsection{Initial condition with matrix inverse}

First consider the initial condition in \cref{eqn:app_1st_ini_cond_w_A}.
In this case, the solution can be represented as 
\begin{equation}
    u(T) = e^{-T A^{-1}} A^{-1} u_0,
\end{equation}
so our goal is to implement the operator $e^{-T A^{-1}} A^{-1}$. 

We start with the Laplace transform
\begin{equation}\label{eqn:app_1st_ODE_laplace}
\frac{1}{z'}e^{-T/z'}=\int_0^\infty e^{-z' t'} J_0(2\sqrt{Tt'}) \ud t'. 
\end{equation}
Here, $J_0$ represents the Bessel function of the first kind of order $0$. 
Since $J_0(2\sqrt{Tt'})$ is not in $L^1$, we cannot directly implement this formula based on LCHS. 
However, since we assume the real part of the matrix is uniformly bounded away from $0$ by $\gamma$, we can consider the shifted version of~\cref{eqn:app_1st_ODE_laplace} by choosing $z' = z + \gamma$ and obtain 
\begin{equation}
\frac{1}{z+\gamma}e^{-T/(z+\gamma)}=\int_0^\infty e^{-zt'} e^{-\gamma t'}J_0(2\sqrt{Tt'}) \ud t'.
\end{equation}
By replacing $z$ by $A-\gamma I = (L - \gamma I) + iH$ and using~\cref{thm:LCHS_hardy}, we have 
\begin{align}
 e^{-T A^{-1}} A^{-1} &= \int_0^{\infty} \int_{\mathbb{R}} \frac{f(k)}{ 1-ik} e^{-\gamma t'} J_0(2\sqrt{Tt'}) e^{-it' (k(L-\gamma I)+H)} \ud k \ud t' \\
&= \int_0^{\infty} \int_{\mathbb{R}} \frac{f(k) {g}(t';T) e^{it' k \gamma} }{ 1-ik} e^{-it' (kL+H)} \ud k \ud t'. \label{eqn:app_1st_ODE_LCHS}
\end{align}
Here 
\begin{equation}
    {g}(t';T)= e^{-\gamma t'}J_0(2\sqrt{Tt'}), 
\end{equation}
and its $L^1$ norm is $\Or(\gamma^{-1})$ and is independent of $T$.

We have represented the operator $e^{-TA^{-1}} A^{-1}$ as LCHS as in~\cref{eqn:app_1st_ODE_LCHS}. 
Notice that~\cref{eqn:app_1st_ODE_LCHS} does not have the standard form specified in~\cref{thm:LCHS_hardy} since we have an extra phase factor $e^{it' k \gamma}$ involving both $t'$ and $k$ and the weight function is not separable in $t'$ and $k$. 
However, we can still implement~\cref{eqn:app_1st_ODE_LCHS} by our algorithm described in~\cref{sec:Lap-LCHS_algorthm}, with a small modification that that we apply an additional controlled phase gate on the select oracle~\cref{eqn:alg_select_oracle} to append the phase $e^{it' k \gamma}$. 
The complexity analysis in~\cref{thm:complexity_block_encoding} and~\cref{cor:complexity_state} still applies with this modification. 

The overall complexity is given in the following result. 
Its proof can be found in~\cref{app:applications_proof_1st_ODE_mass}.

\begin{cor}\label{cor:complexity_app_1st_ODE}
    Consider the differential equation~\cref{eqn:app_1st_ODE} with the initial condition~\cref{eqn:app_1st_ini_cond_w_A}, where the matrix real part of $A$ is positive definite and its eigenvalues are bounded from below by $\gamma > 0$. 
    Then for any fixed $\beta \in (0,1)$, we can prepare an $\epsilon$-approximation of the state $u(T)/\|u(T)\|$ with $\Omega(1)$ success probability and a flag indicating success, using 
    \begin{equation}
        \Or\left( \frac{\norm{u_0}}{\norm{u(T)}} \frac{\alpha_A}{\gamma^2}  \left( \log\left(\frac{\norm{u_0} }{\epsilon \norm{u(T)} \gamma }\right)\right)^{1+1/\beta} \right)
    \end{equation}
    queries to the block encoding of $A$ with block encoding factor $\alpha_A$, and 
    \begin{equation}
        \Or\left( \frac{\norm{u_0}}{\norm{u(T)}} \frac{1}{\gamma} \right)
    \end{equation}
    queries to the state preparation oracle of $\ket{u_0}$. 
    Here, $\beta \in (0,1)$ is the order parameter in the kernel function of the LCHS formula. 
\end{cor}

\subsubsection{Initial condition without matrix inverse}

Now consider the initial condition in \cref{eqn:app_1st_ini_cond_w/o_A}.
In this case, our goal is to implement the operator $e^{-TA^{-1}}$. 
Writing
\begin{equation}
e^{-T/z'}-1=\sum_{n=1}^\infty \frac{(-1)^n}{n!}\frac{T^n}{z'^{n}},
\end{equation}
we can (formally) perform the inverse Laplace transform of each term (the inverse Laplace transform of $z'^{-\alpha}$ is $\frac{1}{\Gamma(\alpha)} t'^{\alpha-1}$), giving
\begin{equation}
\sum_{n=1}^\infty \frac{(-1)^n}{n!}\frac{T^n t'^{n-1}}{\Gamma(n)}=-\sqrt{\frac{T}{t'}}J_1(2\sqrt{Tt'}), 
\end{equation}
where $J_1$ is the Bessel function of the first kind of order 1. 
Therefore (see~\cref{lem:Laplace_1st_DE_mass} in~\cref{app:Laplace_series} for a rigorous proof)
\begin{equation}
e^{-T/z'}-1=-\int_0^\infty e^{-z't'} \sqrt{\frac{T}{t'}}J_1(2\sqrt{Tt'}) \ud t'. 
\end{equation}
By changing the variable $z' = z + \gamma/2$, we have 
\begin{equation}
e^{-T/(z+\gamma/2)}-1=-\int_0^\infty e^{-z t'} e^{-\gamma t'/2} \sqrt{\frac{T}{t'}}J_1(2\sqrt{Tt'}) \ud t'. 
\end{equation}
By replacing $z$ by $L-\gamma I/2 + iH$, which still has a positive definite real part, and using~\cref{thm:LCHS_hardy}, we have 
\begin{align}
    e^{-T A^{-1}}-I &= -\int_0^\infty \int_{\mathbb{R}}  \frac{f(k)}{1-ik} e^{-\gamma t'/2} \sqrt{\frac{T}{t'}}J_1(2\sqrt{Tt'}) e^{-it'(kL-k\gamma I/2 + H)} \ud k \ud t' \\
    &=  \int_0^{\infty} \int_{\mathbb{R}} \frac{f(k) {g}(t';T) e^{it' k \gamma/2} }{ 1-ik} e^{-it' (kL+H)} \ud k \ud t', \label{eqn:app_1st_ODE_LCHS_alg2}
\end{align}
where 
\begin{equation}
    {g}(t';T) = - e^{-\gamma t'/2} \sqrt{\frac{T}{t'}}J_1(2\sqrt{Tt'}). 
\end{equation}
Similarly to the previous case,~\cref{eqn:app_1st_ODE_LCHS_alg2} can also be implemented within our general framework despite an interacting phase factor. 
After obtaining the block encoding of $e^{-TA^{-1}} - I$, we may implement another LCU to add an identity matrix and construct $e^{-TA^{-1}} - I + I = e^{-TA^{-1}}$. 

We give the complexity of this algorithm in the following result. 
Its proof can be found in~\cref{app:applications_proof_1st_ODE_mass_alg2}.

\begin{cor}\label{cor:complexity_app_1st_ODE_alg2}
    Consider the differential equation~\cref{eqn:app_1st_ODE} with initial condition~\cref{eqn:app_1st_ini_cond_w/o_A} where the real part of $A$ is positive definite and its eigenvalues are bounded from below by $\gamma > 0$. 
    Then for any fixed $\beta \in (0,1)$, we can prepare an $\epsilon$-approximation of the state $u(T)/\|u(T)\|$ with $\Omega(1)$ success probability and a flag indicating success, using 
    \begin{equation}
        \Or\left( \frac{\norm{u_0}}{\norm{u(T)}} \frac{\alpha_A}{\gamma} \left(1+\sqrt{\frac{T}{\gamma}}\right) \left( \log\left(\frac{T \norm{u_0} }{\epsilon \norm{u(T)}\gamma } \right) \right)^{1+1/\beta} \right)
    \end{equation}
    queries to the block encoding of $A$ with block encoding factor $\alpha_A$, and 
    \begin{equation}
        \Or\left( \frac{\norm{u_0}}{\norm{u(T)}} \left(1+\sqrt{\frac{T}{\gamma}}\right) \right)
    \end{equation}
    queries to the state preparation oracle of $\ket{u_0}$. 
    Here, $\beta \in (0,1)$ is the order parameter in the kernel function of the LCHS formula. 
\end{cor}

\subsubsection{Comparison with previous algorithms}

Prior to LCHS, the quantum linear differential equation algorithm with lowest query complexity was the truncated Dyson series method~\cite{BerryCosta2022}, which can approximate a quantum state $e^{-T B } \ket{v_0}$ for matrix $B$ with positive semi-definite real part. 
As the input model, this algorithm requires a block encoding of $B$ and a state-preparation oracle for $\ket{v_0}$. 
In our setup, $B = A^{-1}$, and $v_0$ can be either $A^{-1}u_0$ or $u_0$. 
Then, a straightforward approach is to first construct the block encoding of $A^{-1}$ using QSVT and then construct its exponential using the truncated Dyson series method. 

Reference~\cite[Appendix B]{TongAnWiebeEtAl2021} shows that we can construct a block encoding of $A^{-1}$ with block encoding factor $\Or(\norm{A^{-1}})$ using $\Or\left( \alpha_A \norm{A^{-1}} \log\left( {\norm{A^{-1}}}/{\epsilon } \right) \right)$ queries to the block encoding of $A$. 
Then,~\cite[Theorem 1]{BerryCosta2022} shows that preparing an $\epsilon$-approximation of $e^{-T A^{-1}} \ket{v_0} / \norm{e^{-T A^{-1}} \ket{v_0}}$ takes $\widetilde{\Or}\left( \frac{\norm{v_0}}{\norm{u(T)}}\norm{A^{-1}} T  \left(\log(1/\epsilon)\right)^2 \right)$ queries to the block encoding of $A^{-1}$, so the query complexity to the block encoding of $A$ is 
\begin{equation}\label{eqn:app_1st_ODE_Dyson_matrix_query}
    \widetilde{\Or}\left( \frac{\norm{v_0}}{\norm{u(T)}}\alpha_A \norm{A^{-1}}^2 T  \left(\log\left(\frac{1}{\epsilon}\right)\right)^3 \right). 
\end{equation}
We can then replace $\norm{v_0}$ by $\norm{u_0}$ or $\norm{A^{-1}u_0} \leq \norm{A^{-1}}\norm{u_0}$ to obtain the matrix query complexity with different initial conditions. 

Reference~\cite[Theorem 1]{BerryCosta2022} also shows that the algorithm uses $\widetilde{\Or}\left( \frac{\norm{v_0}}{\norm{u(T)}}\norm{A^{-1}} T  \log(1/\epsilon) \right)$ queries to the state preparation oracle for $\ket{v_0}$. 
This is directly the final state preparation cost when $v_0 = u_0$. 
When $v_0 = A^{-1} u_0$, the state $\ket{v_0}$ can be constructed using a quantum linear system solver. 
The optimal algorithm for this~\cite{CostaAnSandersEtAl2022} takes $\Or(\alpha_A \norm{A^{-1}} \log(1/\epsilon) )$ queries to the block encoding of $A$ and the state preparation oracle for $\ket{u_0}$. 
This extra matrix query cost is not dominant compared to~\cref{eqn:app_1st_ODE_Dyson_matrix_query}, but contributes to another multiplicative factor in the state preparation cost, which is $\widetilde{\Or}\left( \frac{\norm{v_0}}{\norm{u(T)}}\alpha_A \norm{A^{-1}}^2 T  \left(\log\left({1}/{\epsilon}\right)\right)^2 \right) = \widetilde{\Or}\left( \frac{\norm{u_0}}{\norm{u(T)}}\alpha_A \norm{A^{-1}}^3 T  \left(\log\left({1}/{\epsilon}\right)\right)^2 \right)$. 

A comparison is given in~\cref{tab:comparison_app_1st_ODE}. 
In both cases, the Lap-LCHS algorithm has better matrix query complexity and state preparation cost. 

\begin{table}[t]
    \renewcommand{\arraystretch}{2}
    \centering
    \scalebox{0.85}{
    \begin{tabular}{c|c|c}\hline\hline
    \multicolumn{3}{c}{ With initial condition $u(0) = A^{-1} u_0$ } \\\hline
        \textbf{Method} & \textbf{Queries to $A$} & \textbf{Queries to $\ket{u_0}$} \\\hline 
        This work (\cref{cor:complexity_app_1st_ODE}) & $\widetilde{\Or}\left( \frac{\norm{u_0}}{\norm{u(T)}} \alpha_A \gamma^{-2} \left( \log\left(\frac{1}{\epsilon}\right)\right)^{1+1/\beta} \right)$ & $\Or\left( \frac{\norm{u_0}}{\norm{u(T)}} \frac{1}{\gamma} \right)$ \\\hline
        QSVT~\cite{GilyenSuLowEtAl2019} + Dyson~\cite{BerryCosta2022} & $\widetilde{\Or}\left( \frac{\norm{u_0}}{\norm{u(T)}}\alpha_A \norm{A^{-1}}^3 T  \left(\log\left(\frac{1}{\epsilon}\right)\right)^3 \right) $ & $\widetilde{\Or}\left( \frac{\norm{u_0}}{\norm{u(T)}}\alpha_A \norm{A^{-1}}^3 T  \left(\log\left(\frac{1}{\epsilon}\right)\right)^2 \right)$
        \\\hline\hline
        \multicolumn{3}{c}{ With initial condition $u(0) = u_0$ } \\\hline
        \textbf{Method} & \textbf{Queries to $A$} & \textbf{Queries to $\ket{u_0}$} \\\hline 
        This work (\cref{cor:complexity_app_1st_ODE_alg2}) & $\widetilde{\Or}\left( \frac{\norm{u_0}}{\norm{u(T)}} \alpha_A \gamma^{-3/2} \sqrt{T} \left( \log\left(\frac{1 }{\epsilon} \right) \right)^{1+1/\beta} \right)$ & $\Or\left( \frac{\norm{u_0}}{\norm{u(T)}} \sqrt{\frac{T}{\gamma}} \right)$ \\\hline
        QSVT~\cite{GilyenSuLowEtAl2019} + Dyson~\cite{BerryCosta2022} & $\widetilde{\Or}\left( \frac{\norm{u_0}}{\norm{u(T)}}\alpha_A \norm{A^{-1}}^2 T  \left(\log\left(\frac{1}{\epsilon}\right)\right)^3 \right) $ & $\widetilde{\Or}\left( \frac{\norm{u_0}}{\norm{u(T)}}\norm{A^{-1}} T  \log\left(\frac{1}{\epsilon}\right) \right)$
        \\\hline\hline
    \end{tabular}
    }
    \caption{ Comparison between Lap-LCHS and the previous approach for solving the ODE~\cref{eqn:app_1st_ODE}. 
    Here, we assume the Cartesian decomposition $A = L+iH$ with $L \succeq \gamma > 0$. $\alpha_A$ is the block encoding factor of $A$, $T$ is the evolution time (and for technical simplicity, we assume $T \geq \gamma$), and $\epsilon$ is the tolerated error in the output state. }
    \label{tab:comparison_app_1st_ODE}
\end{table}

\subsubsection{Application to evolutionary partial differential equations with time-space mixed derivatives}

As we describe in this section, linear differential equations with mass matrices as in~\cref{eqn:app_1st_ODE} can describe certain 
evolutionary partial differential equations with time-space mixed derivatives, so the Lap-LCHS algorithm can be applied to such problems. 
We start with the example 
\begin{align}
    \frac{\partial u(t,x)}{ \partial t} &= \frac{\partial^2 u(t,x)}{\partial t \, \partial x} - u(t,x), \quad t \in [0,T],\, x \in [0,1], \\
    u(0,x) &= u_0(x), \\
    u(t,0) &= u(t,1). 
\end{align}
A standard technique for numerically solving partial differential equations is the method of lines, in which we first discretize all but the time variable to obtain an ODE system and then apply a numerical ODE solver. 
We apply the central difference formula $\partial_x v(t,x) \approx \frac{1}{2h} (v(t,x+h) - v(t,x-h)) $ to discretize the spatial variable $x$ with step size $h$, giving the semi-discretized system 
\begin{equation}\label{eqn:hyperbolic_PDE_mixed_deriv_semi_discretized}
    \frac{\partial u(t,x)}{ \partial t} \approx \frac{1}{2h} \left(\frac{\partial u(t,x+h)}{\partial t } - \frac{\partial u(t,x-h)}{\partial t }\right)- u(t,x). 
\end{equation}
Let $[0,h,2h,\cdots, (N-1)h]$ be the grid points for discretizing $x$, where $N$ is the number of grid points and $h = 1/N$, and $\mathbf{u}(t) = [u(t,0);u(t,h);u(t,2h);\cdots;u(t,(N-1)h)]$. 
Then from~\cref{eqn:hyperbolic_PDE_mixed_deriv_semi_discretized} we have 
\begin{equation}\label{eqn:hyperbolic_PDE_mixed_deriv_ODE}
    \frac{d \mathbf{u}}{dt} \approx D \frac{d \mathbf{u}}{dt} - \mathbf{u}, 
\end{equation}
where 
\begin{equation}
    D = \frac{N}{2} \left( \begin{array}{cccccc}
         0 & 1 & & & & -1 \\
         -1 & 0 & 1 & & & \\
         & -1 & 0 & 1 & & \\
         & & \ddots & \ddots  &\ddots & \\
         & & & -1 & 0 & 1 \\
         1 & & & & -1 & 0 \\
    \end{array} \right). 
\end{equation}
Notice that all the eigenvalues of $D$ are imaginary, 
so~\cref{eqn:hyperbolic_PDE_mixed_deriv_ODE} is exactly in the form of~\cref{eqn:app_1st_ODE} with $A = L + iH$, $L = I \succ 0$ and $H = iD$. 

More generally, we can consider the equation 
\begin{equation}
    \partial_t \mathcal{L}_x u(t,x) = - u(t,x), 
\end{equation}
where $\mathcal{L}_x = \sum_{k=0}^{m} a_k \partial_x^k$ is a differential operator with respect to the $x$ variable. 
Then, the spatially discretized system is also in the form of~\cref{eqn:app_1st_ODE}, where $A$ is the discrete version of $\mathcal{L}_x$. 
The stability condition $L \succ 0$ can be satisfied by imposing conditions on the coefficients $a_{2k}$ of the even-order derivatives. 
It is also possible to generalize this approach to the case where $x$ has higher dimension. 

\subsection{Second-order differential equations}\label{sec:secondorder}

Consider the second-order differential equation
\begin{equation}\label{eqn:app_2nd_ODE}
\frac{\ud^2 u}{\ud t^2}=Au, \quad u(0) = \ket{u_0}, \quad u'(0)=v_0,
\end{equation}
where $A$ has a positive semi-definite real part. 
The solution of \cref{eqn:app_2nd_ODE} has two branches, $e^{\pm \sqrt{A}t}$. 
Here we only consider a special scenario where the choice of $v_0$ forces the dynamics to have only a decaying branch, i.e., the solution is $u(T) = e^{- \sqrt{A} T} \ket{u_0}$. 

In the Lap-LCHS framework, we have 
\begin{equation}\label{eqn:app_2nd_ODE_Laplace}
h(z) = e^{-T\sqrt{z}}=\int_0^\infty e^{-zt'} \frac{T}{2\sqrt{\pi t'^3}} e^{-T^2 / (4t')}
\ud t'.
\end{equation}
Then the inverse Laplace transform of $h(z)$ is 
\begin{equation}\label{eqn:app_2nd_ODE_def_g}
    g(t') = \frac{T}{2\sqrt{\pi t'^3}} e^{-T^2 / (4t')}. 
\end{equation}
We can directly apply the Lap-LCHS algorithm to prepare $h(A)\ket{u_0}$.  Its complexity can be estimated using~\cref{cor:complexity_state}, which is proven in~\cref{app:applications_proof_2nd_ODE}.

\begin{cor}\label{cor:app_2nd_ODE}
    Consider the second-order differential equation in~\cref{eqn:app_2nd_ODE}, where the coefficient matrix $A$ has a positive semi-definite real part. 
    Then for any fixed $\beta \in (0,1)$, Lap-LCHS can prepare an $\epsilon$-approximation of the decaying branch $\ket{u(T)} = u(T)/\norm{u(T)}$, where $u(T) = e^{-\sqrt{A} T} \ket{u_0}$, with $\Omega(1)$ success probability and a flag indicating success, using 
    \begin{equation}
            \Or\left( \frac{\alpha_A T^2}{ \norm{u(T)}^3 \epsilon^2 }  \left( \log\left( \frac{1}{\norm{u(T)}\epsilon} \right) \right)^{1/\beta} \right)
    \end{equation}
    queries to the block encoding of $A$ with block encoding factor $\alpha_A$, and 
    \begin{equation}
        \Or\left( \frac{1}{\norm{u(T)}} \right)
    \end{equation}
    queries to the state preparation oracle for $\ket{u_0}$. 
\end{cor}

\cref{cor:app_2nd_ODE} shows that the matrix query complexity of Lap-LCHS depends badly on several parameters, including $T$, $\epsilon$, and $\norm{u(T)}$. 
This is mainly due to the slow decay of $g(t)$ in~\cref{eqn:app_2nd_ODE_def_g}, which requires a large truncation parameter $T'$. 
It is possible to improve this scaling in a more restricted case where the real part of $A$ is positive definite, and its eigenvalues are lower bounded by $\gamma$, using the same shifting trick as in~\cref{sec:firstorder_mass}. 
Specifically, by substituting the variable $z \to z+\gamma$ in~\cref{eqn:app_2nd_ODE_Laplace}, we have 
\begin{equation}
\frac{1}{z+\gamma}e^{-T\sqrt{z+\gamma}}=\int_0^\infty e^{-zt'} e^{-\gamma t'}\frac{T}{2\sqrt{\pi t'^3}} e^{-T^2 / 4t'}
\ud t'.
\end{equation}
Replacing $z$ by $A-\gamma I=(L-\gamma I) +i H$, we have
\begin{equation}
e^{-T \sqrt{A}}= \int_0^{\infty} \int_{\mathbb{R}} \frac{f(k){g}(t';T)e^{it'k \gamma}}{1-ik} e^{-it' (kL+H)} \ud k \ud t'. 
\end{equation}
Now 
\begin{equation}
    {g}(t';T)= e^{-\gamma t'}\frac{T}{2\sqrt{\pi t'^3}} e^{-T^2 / 4t'}, 
\end{equation}
which has an exponential decay as $t'\to \infty$, so the truncation parameter can be significantly reduced to $\Or( (1/\gamma) \log(1/\epsilon'))$ in order to bound $\norm{g}_{L^1((T',\infty))}$ by $\Or(\epsilon')$. 
The complexity of Lap-LCHS in this special scenario is captured by the following result, which can be directly obtained from the proof of~\cref{cor:app_2nd_ODE} and the new choice of $T'$. 

\begin{cor}\label{cor:app_2nd_ODE_PD}
    Consider the second-order differential equation in~\cref{eqn:app_2nd_ODE}, where the coefficient matrix $A$ has positive definite matrix real part $L \succeq \gamma > 0$. 
    Then for any fixed $\beta \in (0,1)$, Lap-LCHS can prepare an $\epsilon$-approximation of the decaying branch $\ket{u(T)} = u(T)/\norm{u(T)}$, where $u(T) = e^{-\sqrt{A} T} \ket{u_0}$, with $\Omega(1)$ success probability and a flag indicating success, using 
    \begin{equation}
            \Or\left( \frac{1}{\norm{u(T)} } \frac{\alpha_A}{\gamma} \left( \log\left( \frac{1}{\norm{u(T)}\epsilon} \right) \right)^{1+1/\beta}  \right)
    \end{equation}
    queries to the block encoding of $A$ with block encoding factor $\alpha_A$, and 
    \begin{equation}
        \Or\left( \frac{1}{\norm{u(T)}} \right)
    \end{equation}
    queries to the state preparation oracle for $\ket{u_0}$. 
\end{cor}

\subsubsection{Second-order differential equations with mass matrices}\label{sec:secondorder_mass}

Analogous to the first-order case, we may also consider a second-order differential equation with a non-diagonal mass matrix
\begin{equation}\label{eqn:app_2nd_ODE_mass}
A\frac{\ud^2 u}{\ud t^2}=u, \quad u(0) = \ket{u_0}, \quad u'(0)=v_0.
\end{equation}
We focus on the decaying branch of the solution $u(T) = e^{-\sqrt{A^{-1}}T}\ket{u_0} $. Here, we briefly discuss the simulation strategy and omit the detailed complexity analysis for simplicity.
Using the fact that
\begin{equation}
e^{-T/\sqrt{z}}-1=\sum_{n=1}^\infty \frac{(-1)^n}{n!}\frac{T^n}{z^{n/2}},
\end{equation}
we can perform the inverse Laplace transform of each term (in particular, the inverse Laplace transform of $z^{-\alpha}$ is $\frac{1}{\Gamma(\alpha)} t'^{\alpha-1}$) and obtain 
(see \cref{lem:Laplace_2nd_DE_mass} in~\cref{app:Laplace_series})
\begin{equation}
e^{-T/\sqrt{z}}-1 = \int_0^{\infty} e^{-zt'} {g}(t') \ud t'
\end{equation}
where
\begin{equation}
{g}(t')=\sum_{n=1}^\infty \frac{(-1)^n T^n}{n!}\frac{t'^{n/2-1}}{\Gamma(n/2)}=
\frac{1}{2} T^2 \times{}{}_0F_2\left(;\frac{3}{2},2; \frac{T^2 t'}{4} \right) - \frac{T}{\sqrt{\pi t'}}\times {}_0F_2\left(;\frac{1}{2}, \frac{3}{2}; \frac{T^2 t'}{4} \right), 
\end{equation}
with ${}_pF_q(a_1, \ldots, a_p; b_1, \ldots, b_q; z)$ denoting the generalized hypergeometric function.

Suppose that the real part of $A$ is positive definite and its eigenvalues are lower bounded by $\gamma > 0$. 
Using the shifting trick $z \to z + \gamma$ again, we have 
\begin{equation}
e^{-T/\sqrt{z+\gamma}}-1 = \int_0^{\infty} e^{-zt'} e^{-\gamma t'} {g}(t') \ud t'. 
\end{equation}
Then, by replacing $z$ by $A - \gamma I$, we have 
\begin{equation}
e^{-T/\sqrt{A}}= \int_0^{\infty} \int_{\mathbb{R}} \frac{f(k){g}(t';T)e^{it'k \gamma}}{1-ik} e^{-it' (kL+H)} \ud k \ud t',  
\end{equation}
with 
\begin{equation}
    {g}(t';T) = e^{-\gamma t'} \left( \sum_{n=1}^\infty \frac{(-1)^n T^n}{n!}\frac{t'^{n/2-1}}{\Gamma(n/2)} \right). 
\end{equation}
So $e^{-T/\sqrt{A}}$ can be implemented by the Lap-LCHS algorithm.

\bibliographystyle{alpha}
\bibliography{reference}

\clearpage

\appendix

\section{Bounding the integral discretization error}\label{app:quadrature_error}

\begin{proof}[Proof of~\cref{lem:quadrature}]
    We focus on bounding two types of errors. 
    Choices of the parameters can be directly derived from the error bounds. 
    For the truncation error, we have 
    \begin{align}
        & \quad \norm{ \int_0^{\infty} \int_{\mathbb{R}} \frac{f(k){g}(t)}{ 1-ik} U(k,t) \ud k \ud t -  \int_0^{T} \int_{-K}^K \frac{f(k){g}(t)}{ 1-ik} U(k,t) \ud k \ud t } \\
        & \leq \norm{ \int_0^{\infty} \int_{\mathbb{R}} \frac{f(k){g}(t)}{ 1-ik} U(k,t) \ud k \ud t -  \int_0^{T} \int_{\mathbb{R}} \frac{f(k){g}(t)}{ 1-ik} U(k,t) \ud k \ud t } \\
        & \quad\quad + \norm{ \int_0^{T} \int_{\mathbb{R}} \frac{f(k){g}(t)}{ 1-ik} U(k,t) \ud k \ud t -  \int_0^{T} \int_{-K}^K \frac{f(k){g}(t)}{ 1-ik} U(k,t) \ud k \ud t } \\
        & = \norm{ \int_T^{\infty} \int_{\mathbb{R}} \frac{f(k){g}(t)}{ 1-ik} U(k,t) \ud k \ud t} +  \norm{ \int_0^{T} \int_{\mathbb{R}\setminus [-K,K]} \frac{f(k){g}(t)}{ 1-ik} U(k,t) \ud k \ud t } \\
        & \leq \norm{\frac{f}{1-ik}}_{L^1(\mathbb{R})} \norm{{g}}_{L^1((T,\infty))} + \norm{\frac{f}{1-ik}}_{L^1(\mathbb{R}\setminus[-K.K])} \norm{{g}}_{L^1(\RR_+)}\\
        & \leq \norm{f}_{L^1(\mathbb{R})} \norm{{g}}_{L^1((T,\infty))} + \norm{f}_{L^1(\mathbb{R}\setminus[-K.K])} \norm{{g}}_{L^1(\RR_+)}. 
    \end{align}

    For the quadrature error, let $V(k,j) = \frac{f(k){g}(t)}{ 1-ik} U(k,t) $. 
    We first write 
    \begin{align}
         & \quad \int_0^{T} \int_{-K}^K V(k,t) \ud k \ud t - h_k h_t \sum_{j=0}^{2K/h_k-1}\sum_{l=0}^{T/h_t-1} V(k_j,t_l) \\
         & = \sum_{j=0}^{2K/h_k-1}\sum_{l=0}^{T/h_t-1} \left( \int_{t_l}^{t_l+h_t} \int_{k_j}^{k_j+h_j} V(k,t) \ud k \ud t - h_k h_t V(k_j,t_l) \right). 
    \end{align}
    On each interval $[k_j,k_j+h_j]\times [t_l,t_l+h_t]$, we have 
    \begin{equation}
        \norm{V(k,t)-V(k_j,t_l)} \leq \sup \norm{\frac{\partial V}{\partial k}} h_k +  \sup \norm{\frac{\partial V}{\partial t}} h_t, 
    \end{equation}
    so 
    \begin{align}
         & \quad \norm{ \int_0^{T} \int_{-K}^K V(k,t) \ud k \ud t - h_k h_t \sum_{j=0}^{2K/h_k-1}\sum_{l=0}^{T/h_t-1} V(k_j,t_l)} \\
         & \leq \sum_{j=0}^{2K/h_k-1}\sum_{l=0}^{T/h_t-1} h_kh_t \left( \sup \norm{\frac{\partial V}{\partial k}} h_k +  \sup \norm{\frac{\partial V}{\partial t}} h_t \right) \\
         & = 2 K T \left( \sup \norm{\frac{\partial V}{\partial k}} h_k +  \sup \norm{\frac{\partial V}{\partial t}} h_t \right). 
    \end{align}
    Straightforward computations give 
    \begin{align}
        \frac{\partial V}{\partial k} = \frac{f(k){g}(t)}{1-ik} (-itL) e^{-it(kL+H)} + \frac{f'(k)(1-ik)+if(k)}{(1-ik)^2} {g}(t) e^{-it(kL+H)}, 
    \end{align}
    \begin{align}
        \frac{\partial V}{\partial t} = \frac{f(k){g}(t)}{1-ik} (-i)(kL+H) e^{-it(kL+H)} + \frac{f(k) {g}'(t)}{1-ik} e^{-it(kL+H)}, 
    \end{align}
    and thus 
    \begin{align}
        \sup \norm{\frac{\partial V}{\partial k}} \leq \norm{f}_{L^{\infty}(\mathbb{R})} \norm{t {g}}_{L^{\infty}(\mathbb{R})} \norm{A} + \left( \norm{f'}_{L^{\infty}(\mathbb{R})} + \norm{f}_{L^{\infty}(\mathbb{R})} \right) \norm{{g}}_{L^{\infty}(\mathbb{R})}, 
    \end{align}
    \begin{align}
        \sup \norm{\frac{\partial V}{\partial t}} \leq 2\norm{f}_{L^{\infty}(\mathbb{R})} \norm{{g}}_{L^{\infty}(\mathbb{R})} \norm{A} + \norm{f}_{L^{\infty}(\mathbb{R})} \norm{{g}'}_{L^{\infty}(\mathbb{R})}. 
    \end{align}
\end{proof}

\section{Construction of the prepare oracles}\label{app:prepare_oracle}

Here we discuss the construction of the unitaries $O_{c,l}$, $O_{c,r}$, $O_{\hat{c},l}$, and $O_{\hat{c},r}$, which prepare quantum states encoding the quadrature coefficients in their amplitudes. 
Specifically, these unitaries prepare a quantum state of the form
\begin{equation}\label{eqn:def_general_state_prep}
    \frac{1}{\sqrt{\sum_{j=0}^{M-1} |F(a+jh)|^2 }}\sum_{j=0}^{M-1} F(a+jh) \ket{j}. 
\end{equation}
Here $F(k)$ is a known complex-valued function defined on a real interval $[a,b]$, $M$ is a positive integer, and $h = (b-a)/M$ is the step size. 

In general, an $M$-dimensional quantum state can be prepared with cost $\Or(M)$ \cite{SBM06}.
However, the amplitudes of the quantum state in~\cref{eqn:def_general_state_prep} are given by known functions evaluated at discrete points, and these functions in the Lap-LCHS algorithm have closed-form expressions and are integrable (e.g., in $O_{c,r}$ we have $F(k) = \frac{f(k)}{1-ik}$ where $f(k)$ is the kernel function in LCHS and can be chosen as in~\cref{eqn:kernel_intro}). 
Then, the state preparation circuits can be constructed more efficiently, with cost only $\Or(\poly\log M)$, by the Grover-Rudolph approach~\cite{GroverRudolph2002}. 
More specifically, we first apply Grover-Rudolph to prepare a state proportional to $\sum_{j=0}^{M-1} |F(a+jh)| \ket{j}$ encoding the absolute values of $F(k)$ in its amplitudes, and then apply the unitary $\widetilde{U}: \ket{j} \rightarrow e^{i \arg (F(a+jh)) } \ket{j} $ to introduce the correct phases. 
Notice that $\widetilde{U}$ can be efficiently constructed with the help of classical arithmetic, and the overall complexity of preparing the desired state is $\Or(\poly\log M)$. 

The implementation of the Grover-Rudolph approach requires coherent arithmetic, resulting in complicated quantum circuits due to its high qubit and gate costs in practice. 
A recent work~\cite{McardleGilyenBerta2022} provides an alternative state preparation algorithm based on QSVT and avoids coherent arithmetic operations, provided the function $F(k)$ can be approximated by a degree-$d$ polynomial or Fourier series (which is also the case in Lap-LCHS). 
As shown in~\cite[Theorem 1]{McardleGilyenBerta2022}, this approach only needs $4$ ancilla qubits and $\Or(\log(M) d / \mathcal{F})$ gates, where $\mathcal{F} = \sqrt{ \frac{ \sum_{j=0}^{M-1} |F(a+jh)|^2 }{M \max_j |F(a+jh)|^2} }$. 
Notice that if the function $F(k)$ performs badly in the sense that it concentrates only on a few points, e.g., $F(a+jh)$ is non-zero only at $j = 0$, then $\mathcal{F} \sim \sqrt{1/M}$ and such a state preparation approach still has cost $\widetilde{\Or}(\sqrt{M})$. 
However, if the function $F(k)$ is continuously differentiable and in $L^1$ (which is satisfied for $O_{c,l}$ and $O_{c,r}$ and also holds for $O_{\hat{c},l}$ and $O_{\hat{c},r}$ for continuously differentiable inverse Laplace transforms), then $\mathcal{F} \approx \frac{1}{\sqrt{b-a}}\frac{\norm{F}_{L^2(a,b)}}{\norm{F}_{L^{\infty}(a,b)}}$ does not depend on the number $M$, so the gate complexity of the state preparation is $\Or(\log(M))$.

\section{Laplace transform of series of functions}\label{app:Laplace_series}

In~\cref{sec:applications}, we compute the inverse Laplace transforms of some functions by formally taking the inverse Laplace transform of each term in its series expansion. 
Here, we rigorously validate such a procedure by analyzing the term-by-term computation of the Laplace transform. 
This also validates the same procedure for the inverse Laplace transform by definition of the inverse Laplace transform. 

We first state a general convergence result, which is a direct consequence of the dominated convergence theorem. 

\begin{lem}\label{lem:Laplace_series}
    Let $g(t) = \sum_{j=1}^{\infty} a_j q_j(t)$ be a pointwise absolutely convergent function series on $(0,\infty)$. 
    Suppose that $\lim_{J \to \infty} \sum_{j=1}^{J} \int_0^{\infty} |a_j q_j(t) e^{-zt} | \ud t < \infty $ for $z \in \mathcal{D} \subset \overline{\mathbb{C}_{>}}$. Then the Laplace transform of $g(t)$ exists on $\mathcal{D}$ and is given as 
    \begin{equation}
        h(z) = \sum_{j=1}^{\infty} a_j \int_0^{\infty} q_j(t) e^{-zt} \ud t. 
    \end{equation}
\end{lem}
\begin{proof}
    We fix the complex number $z$ in the domain of interest and drop the explicit dependence on it. 
    Let $g_J(t) = \sum_{j=1}^{J} a_j q_j(t)$ and $G(t) = \sum_{j=1}^{\infty} |a_j q_j(t) e^{-zt}| $, which is well-defined due to the absolute convergence. 
    Then $|g_J(t) e^{-zt} | \leq G(t) $, and by Beppo Levi's lemma we have 
    \begin{equation}
        \int_0^{\infty} G(t) \ud t = \lim_{J \to \infty}  \sum_{j=1}^{J} \int_0^{\infty} |a_j q_j(t) e^{-zt}| \ud t < \infty. 
    \end{equation}
    Therefore, by the dominated convergence theorem, we have that $g(t) e^{-zt}$ is integrable and 
    \begin{equation}
        \int_0^{\infty} g(t) e^{-zt} \ud t =  \int_0^{\infty} \lim_{J \to \infty} g_J(t) e^{-zt} \ud t = \lim_{J \to \infty} \int_0^{\infty} g_J(t) e^{-zt} \ud t = \sum_{j=1}^{\infty} a_j \int_0^{\infty} q_j(t) e^{-zt} \ud t, 
    \end{equation}
    where the last step is due to integrability of each $a_j q_j(t) e^{-zt}$ and linearity of the integral. 
\end{proof}

Now, we can give a rigorous proof of the examples of the Laplace transform series we use in~\cref{sec:applications}. 

\begin{lem}\label{lem:Laplace_1st_DE_mass}
    The Laplace transform of $g(t) = \sum_{n=1}^\infty \frac{(-1)^n}{n!}\frac{T^n t^{n-1}}{\Gamma(n)}$ on $\mathbb{C}_{>}$ is $h(z) = e^{-T/z}  - 1$. 
\end{lem}
\begin{proof}
    Notice that the Laplace transform of $t^{n-1}$ is $\Gamma(n) z^{-n}$. 
    We have 
    \begin{equation}
        \lim_{N \to \infty} \sum_{n=1}^N \int_0^{\infty} \left|\frac{(-1)^n}{n!}\frac{T^n t^{n-1}}{\Gamma(n)} e^{-zt} \right| \ud t = \lim_{N \to \infty} \sum_{n=1}^N \frac{1}{n!} \frac{T^n}{(\Re(z) )^n} = e^{T/\Re (z)} - 1 < \infty. 
    \end{equation}
    Then, by~\cref{lem:Laplace_series}, we can compute 
    \begin{equation}
        \int_0^{\infty} g(t) e^{-zt} \ud t = \sum_{n=1}^\infty \frac{(-1)^n}{n!}\frac{T^n }{\Gamma(n)} \int_0^{\infty} t^{n-1} e^{-zt} \ud t = \sum_{n=1}^\infty \frac{(-1)^n}{n!}\frac{T^n }{z^n} = h(z)
    \end{equation}
    as claimed.
\end{proof}

\begin{lem}\label{lem:Laplace_2nd_DE_mass}
    The Laplace transform of $g(t) = \sum_{n=1}^\infty \frac{(-1)^n T^n}{n!}\frac{t^{n/2-1}}{\Gamma(n/2)}$ on $\mathbb{C}_{>}$ is $h(z) = e^{-T/\sqrt{z}}  - 1$. 
\end{lem}
\begin{proof}
    Notice that the Laplace transform of $t^{\alpha-1}$ is $\Gamma(\alpha) z^{-\alpha}$ for $\alpha > 0$. 
    We have 
    \begin{equation}
        \lim_{N\to\infty} \sum_{n=1}^N \int_0^{\infty} \left|\frac{(-1)^n T^n}{n!}\frac{t^{n/2-1}}{\Gamma(n/2)} e^{-zt}\right| \ud t = \lim_{N\to\infty} \sum_{n=1}^N \frac{1}{n!} \frac{T^n}{(\Re (z))^{n/2}} = e^{T/\sqrt{\Re(z)}} - 1 < \infty. 
    \end{equation}
    Then, by~\cref{lem:Laplace_series}, we can compute 
    \begin{equation}
        \int_0^{\infty} g(t) e^{-zt} \ud t = \sum_{n=1}^\infty \frac{(-1)^n T^n}{n!}\frac{1}{\Gamma(n/2)} \int_0^{\infty} t^{n/2-1} e^{-zt} \ud t = \sum_{n=1}^\infty \frac{(-1)^n T^n}{n!}\frac{1}{z^{n/2}} = h(z)
    \end{equation}
    as claimed.
\end{proof}

\section{Complexity estimates for applications of Lap-LCHS}\label{app:applications_proof}

Here, we present the proofs of the complexity estimates for the various applications discussed in~\cref{sec:applications}. 
The ideas of the proofs are similar. 
In each specific application, we compute the $L^1$ norm of the functions $f(k)$ and $g(t)$ and estimate the integral truncation parameter $K$ and $T$ to satisfy the assumption in~\cref{thm:complexity_block_encoding}. 
Then, the claimed results directly follow from~\cref{thm:complexity_block_encoding} or~\cref{cor:complexity_state}. 

\subsection{Proof of \texorpdfstring{\cref{cor:app_linear_system}}{Corollary 9}}\label{app:applications_proof_linear_system}

\begin{proof}[Proof of~\cref{cor:app_linear_system}]
    In this application, we have $\norm{f}_{L^1(\RR)} = \Or(1)$ as usual. 
    The $L^1$ norm of $g(t)$ is 
    \begin{align}
        \norm{g}_{L^1(\RR_+)} = \frac{1}{\Gamma(p)} \int_0^{\infty} t^{p-1} e^{-\eta t} \ud t = \eta^{-p}. 
    \end{align}
    To bound $\norm{f}_{L^1(\RR\setminus[-K,K])}$ by $\Or(\epsilon \norm{x}/\norm{g}_{L^1(\RR_+)}) = \Or(  \epsilon \norm{x} \eta^p )$, we can choose $K = \Or(( \log(1/(\epsilon \eta \norm{x} )) )^{1/\beta})$. 
    We still need to choose $T$ such that $\norm{{g}}_{L^1((T,\infty))} = \Or(\epsilon \norm{x}/\norm{f}_{L^1(\mathbb{R})}) = \Or(\epsilon\norm{x})$. 
    Notice that 
    \begin{align}
        \norm{{g}}_{L^1((T,\infty))} &= \frac{1}{\Gamma(p)} \int_T^{\infty} t^{p-1} e^{-\eta t} \ud t \\
        & \leq  e^{-\eta T /2} \left(\frac{1}{\Gamma(p)} \int_T^{\infty} t^{p-1} e^{-\eta t/2} \ud t\right) \\
        & \leq e^{-\eta T /2} \left(\frac{1}{\Gamma(p)} \int_0^{\infty} t^{p-1} e^{-\eta t/2} \ud t\right) \\
        & = e^{-\eta T /2} (\eta/2)^{-p}. 
    \end{align}
    Then, it suffices to choose 
    \begin{equation}
        T = \Or\left( \frac{1}{\eta} \log\left( \frac{1}{\epsilon \eta \norm{x}} \right) \right). 
    \end{equation}
    Plugging these estimates back to~\cref{cor:complexity_state} gives the claimed query complexity. 
\end{proof}

\subsection{Proof of \texorpdfstring{\cref{cor:complexity_app_1st_ODE}}{Corollary 10}}\label{app:applications_proof_1st_ODE_mass}

\begin{proof}[Proof of~\cref{cor:complexity_app_1st_ODE}]
    We first discuss the complexity of implementing a block encoding of $e^{-TA^{-1}}A^{-1} $. 
    In its LCHS representation, we have 
    \begin{equation}
        f(k) = \frac{1}{2\pi e^{-2^{\beta}} e^{(1+ik)^{\beta}} },  
    \end{equation}
    and 
    \begin{equation}
        {g}(t';T)= e^{-\gamma t'}J_0(2\sqrt{Tt'}). 
    \end{equation}
    Notice that 
    \begin{equation}
        |{g}(t';T)| \leq  e^{-\gamma t'}. 
    \end{equation}
    We can estimate the $L^1$ norms of the functions as 
    \begin{equation}
        \norm{f}_{L^1} = \Or(1), \quad \norm{{g}}_{L^1} = \Or(1/\gamma),  
    \end{equation}
    and choose 
    \begin{equation}
        T' = \Or\left( \frac{1}{\gamma} \log\left( \frac{1}{\gamma \epsilon'} \right) \right)  , \quad K = \Or\left(  \left( \log\left(\frac{1}{\gamma \epsilon'}\right)\right)^{1/\beta} \right)
    \end{equation}
    such that $\norm{{g}}_{L^1([T',\infty])} = \Or(\epsilon'/\norm{f}_{L^1(\mathbb{R})})$ and $\norm{f}_{L^1(\mathbb{R}\setminus[-K.K])}  = \Or(\epsilon'/\norm{{g}}_{L^1(0,\infty)})$. 
    According to~\cref{thm:complexity_block_encoding}, we can construct an $(\alpha'=\Or(1/\gamma),\epsilon')$-block encoding of $ e^{-T A^{-1}} A^{-1}$ using 
    \begin{equation}\label{eqn:app_1st_ODE_proof_eq1}
        \Or\left( \frac{\alpha_A}{\gamma} \left( \log\left(\frac{1}{\gamma \epsilon'}\right)\right)^{1+1/\beta} + \log\left(\frac{1}{\gamma\epsilon'}\right)\right) =  \Or\left( \frac{\alpha_A}{\gamma} \left( \log\left(\frac{1}{\gamma \epsilon'}\right)\right)^{1+1/\beta} \right)
    \end{equation}
    queries to the block encoding of $A$. 

    Applying this block encoding to the input state $\ket{0}\ket{u_0}$ gives 
    \begin{equation}
        \frac{1}{\alpha'} \ket{0} B \ket{u_0} + \ket{\perp}, 
    \end{equation}
    where $B$ is a matrix such that $\norm{B - e^{-TA^{-1}}A^{-1}} \leq \epsilon'$. 
    Measuring the ancilla register onto $0$ gives the state $B\ket{u_0}/\norm{B\ket{u_0}}$. This is an approximation of the desired state 
    \begin{equation}
        \frac{e^{-TA^{-1}}A^{-1}\ket{u_0}}{\norm{e^{-TA^{-1}}A^{-1}\ket{u_0}}}
    \end{equation}
    up to error 
    \begin{equation}
        \norm{ \frac{B\ket{u_0}}{ \norm{B\ket{u_0}} } - \frac{ e^{-TA^{-1}}A^{-1}\ket{u_0} }{ \norm{e^{-TA^{-1}}A^{-1}\ket{u_0}} }  } \leq \frac{2}{ \norm{e^{-TA^{-1}}A^{-1}\ket{u_0}} } \norm{B - e^{-TA^{-1}}A^{-1}} \leq \frac{2 \epsilon' \norm{u_0} }{ \norm{u(T)} }. 
    \end{equation}
    To bound the final error by $\epsilon$ and the success probability by $\Omega(1)$, it suffices to choose 
    \begin{equation}
        \epsilon' = \Or\left( \frac{ \norm{u(T)} }{ \norm{u_0} } \epsilon \right) 
    \end{equation}
    and implement $\Or(\alpha'/\norm{B\ket{u_0}}) = \Or\left(\frac{1}{\gamma} \frac{\norm{u_0}}{\norm{u(T)}}\right)$ rounds of amplitude amplification. 
    Plugging these choices back into~\cref{eqn:app_1st_ODE_proof_eq1} gives the overall query complexity. 
\end{proof}

\subsection{Proof of \texorpdfstring{\cref{cor:complexity_app_1st_ODE_alg2}}{Corollary 11}}\label{app:applications_proof_1st_ODE_mass_alg2}

\begin{proof}[Proof of~\cref{cor:complexity_app_1st_ODE_alg2}]
    We start by estimating the complexity of block-encoding $e^{-TA^{-1}} - I$ based on~\cref{eqn:app_1st_ODE_LCHS_alg2} by estimating the $L^1$ norm and truncation parameter of the functions 
    \begin{equation}
        f(k) = \frac{1}{2\pi e^{-2^{\beta}} e^{(1+ik)^{\beta}} }  
    \end{equation}
    and 
    \begin{equation}
        {g}(t';T) = - e^{-\gamma t'/2} \sqrt{\frac{T}{t'}}J_1(2\sqrt{Tt'}). 
    \end{equation}
    We have $\norm{f}_{L^1} = \Or(1)$ as usual. 
    By bounding the Bessel function by $1$ and substituting the variable $s = \sqrt{Tt'}$, we have 
    \begin{align}
        \norm{{g}}_{L^1} &\leq \int_0^{\infty} e^{-\gamma t'/2} \sqrt{\frac{T}{t'}} \ud t' \\
        &= 2 \int_0^{\infty} e^{- \frac{s^2}{2T/\gamma} } \ud s \\
        & = \Or\left( \sqrt{\frac{T}{\gamma}} \right). 
    \end{align}
    We choose $K$ and $T'$ so that $\norm{{g}}_{L^1([T',\infty])} = \Or(\epsilon'/\norm{f}_{L^1(\mathbb{R})})$ and $\norm{f}_{L^1(\mathbb{R}\setminus[-K.K])}  = \Or(\epsilon'/\norm{{g}}_{L^1(0,\infty)})$. 
    Using similar techniques as in estimating $\norm{{g}}_{L^1}$, we obtain 
    \begin{align}
    \norm{{g}}_{L^1([T',\infty])} & \leq \int_{T'}^{\infty} e^{-\gamma t'/2} \sqrt{\frac{T}{t'}} \ud t' \\
        &= 2 \int_{\sqrt{T T'}}^{\infty} e^{- \frac{s^2}{2T/\gamma} } \ud s \\ 
        & \leq 2 e^{- \frac{TT'}{4T/\gamma} } \int_{0}^{\infty} e^{- \frac{s^2}{4T/\gamma} } \ud s \\
        & = \Or\left( e^{-\gamma T'/4 } \sqrt{\frac{T}{\gamma}}  \right). 
    \end{align}
    Therefore, it suffices to choose 
    \begin{equation}
        T' = \Or\left( \frac{1}{\gamma} \log\left( \frac{T}{\epsilon' \gamma} \right) \right), \quad K = \Or\left( \left( \log\left(\frac{T}{\epsilon' \gamma} \right) \right)^{1/\beta} \right). 
    \end{equation}
    According to~\cref{thm:complexity_block_encoding}, we can construct an $(\Or(\sqrt{T/\gamma}), \epsilon')$-block encoding of $e^{-TA^{-1}} - I$ using 
    \begin{equation}\label{eqn:proof_app_1st_ODE_alg2_eq1}
        \Or\left( \frac{\alpha_A}{\gamma} \left( \log\left(\frac{T}{\epsilon' \gamma} \right) \right)^{1+1/\beta} + \log\left( \frac{T}{\epsilon'\gamma }\right) \right) = \Or\left( \frac{\alpha_A}{\gamma} \left( \log\left(\frac{T}{\epsilon' \gamma} \right) \right)^{1+1/\beta} \right)
    \end{equation} 
    queries to the block encoding of $A$.

    Now the block encoding of $e^{-TA^{-1}}$ can be directly constructed by using LCU for $e^{-TA^{-1}} - I$ and $I$. 
    According to~\cite[Lemma 52]{GilyenSuLowEtAl2019}, this is an $(\Or(1+\sqrt{T/\gamma}), \epsilon')$-block encoding with the same query complexity. 
    Using the same analysis in the proof of~\cref{cor:complexity_app_1st_ODE}, to bound the error in the final output state by $\epsilon$ and the failure probability by $1-\Omega(1)$, it suffices to choose 
    \begin{equation}
        \epsilon' = \Or\left( \frac{\norm{u(T)}}{\norm{u_0}} \epsilon \right)
    \end{equation}
    and run $\Or\left( \left(1+ \sqrt{T/\gamma} \right)\frac{\norm{u_0}}{\norm{u(T)}} \right)$ rounds of amplitude amplification. 
    Plugging these estimates back into~\cref{eqn:proof_app_1st_ODE_alg2_eq1} yields the desired complexity estimates. 
\end{proof}

\subsection{Proof of \texorpdfstring{\cref{cor:app_2nd_ODE}}{Corollary 12}}\label{app:applications_proof_2nd_ODE}

\begin{proof}[Proof of~\cref{cor:app_2nd_ODE}]
    To use~\cref{cor:complexity_state}, we only need to estimate the $L^1$ norms of $f$ and $g$ and determine the integral truncation parameters. 
    We have $\norm{f}_{L^1(\RR)} = \Or(1)$ as usual, and by substituting the variable $s = T^2/(4t')$, we have 
    \begin{equation}
        \norm{g}_{L^1(\RR_+)} = \int_0^{\infty} \frac{T}{2\sqrt{\pi t'^3}} e^{-T^2 / (4t')} \ud t' = \int_0^{\infty} \frac{1}{\sqrt{\pi s}} e^{-s} \ud s = \Or(1). 
    \end{equation}
    Then we choose $K$ and $T'$ such that $\norm{{g}}_{L^1((T',\infty))} = \Or(\epsilon \norm{u(T)})$ and $\norm{f}_{L^1(\mathbb{R}\setminus[-K,K])}  = \Or(\epsilon \norm{u(T)})$. 
    This can be achieved by choosing $K = \Or( (\log(1/(\norm{u(T)}\epsilon)  ))^{1/\beta} )$ as usual, and we take $T' = \Or(T^2/(\epsilon\norm{u(T)})^2)$ because 
    \begin{align}
        \norm{{g}}_{L^1((T',\infty))} = \int_{T'}^{\infty} \frac{T}{2\sqrt{\pi t'^3}} e^{-T^2 / (4t')} \ud t' 
         \leq \int_{T'}^{\infty} \frac{T}{2\sqrt{\pi t'^3}}  \ud t' = \frac{T}{\sqrt{\pi T'}}. 
    \end{align}
\end{proof}

\end{document}